\documentclass[12pt]{article}
\usepackage{amsmath}
\usepackage{makecell}
\usepackage{hyperref}
\usepackage{amsfonts}
\usepackage[utf8]{inputenc}
\usepackage{graphicx,psfrag,epsf}
\usepackage{subcaption}
\usepackage{amsthm}
\usepackage{float}
\usepackage[space]{grffile}
\usepackage{mathrsfs}
\usepackage{bbm}
\usepackage{mathtools}
\usepackage[left=1.1in,top=1.40in,right=1.1in,bottom=1.6in,head=1.40in]{geometry}
\usepackage[normalem]{ulem}

\usepackage{sectsty}
\sectionfont{\centering}
\usepackage{titlesec}
\titlelabel{\thetitle.\quad}

\usepackage{natbib}
\usepackage[english]{babel}

\setcitestyle{aysep={}}

\newcommand{\ybold}{\boldsymbol{y}}
\newcommand{\xbold}{\boldsymbol{x}}
\newcommand{\sbold}{\boldsymbol{s}}
\newcommand{\betabold}{\boldsymbol{\beta}}
\newcommand{\epsilonbold}{\boldsymbol{\epsilon}}
\newcommand{\zbold}{\boldsymbol{z}}
\newcommand{\bbold}{\boldsymbol{b}}
\newcommand{\phibold}{\boldsymbol{\phi}}
\newcommand{\hbold}{\boldsymbol{h}}

\newcommand{\E}{\mathbb{E}}

\newcommand{\Cov}{\mathrm{Cov}}
\newcommand{\tr}{\mathrm{tr}}
\newcommand{\argmin}{\mathrm{argmin}}
\newcommand\independent{\protect\mathpalette{\protect\independenT}{\perp}}
\def\independenT#1#2{\mathrel{\rlap{$#1#2$}\mkern2mu{#1#2}}}

\newtheorem{theorem}{Theorem}[section]

\newtheorem{definition}{Definition}[section]
\newtheorem{remark}{Remark}[section]

\hypersetup{
    colorlinks,
    citecolor=black,
    filecolor=blue,
    linkcolor=blue,
    urlcolor=blue
}

\newcommand{\blind}{1}

\title{Cross Validation for Correlated Data}
\date{April 05, 2020}

\begin{document}
\def\spacingset#1{\renewcommand{\baselinestretch}%
{#1}\small\normalsize} \spacingset{1}


\if1\blind
{
  \title{\bf Cross-Validation for Correlated Data}
  \author{Assaf Rabinowicz \thanks{
    The authors gratefully acknowledge \textit{Israel Science Foundation, grant 1804/16}}\hspace{.2cm}\\
    Department of Statistics, Tel-Aviv University, Tel-Aviv, Israel, 69978\\
    \\
    Saharon Rosset\\
    Department of Statistics, Tel-Aviv University, Tel-Aviv, Israel, 69978}
  \maketitle
} \fi

\if0\blind
{
  \bigskip
  \bigskip
  \bigskip
  \begin{center}
    {\LARGE\bf Cross-Validation for Correlated Data}
\end{center}
  \medskip
} \fi

\bigskip


\begin{abstract}
K-fold cross-validation (CV) with squared error loss is widely used for evaluating predictive models, especially when strong distributional assumptions cannot be taken. However, CV with squared error loss is not free from distributional assumptions, in particular in cases involving non-i.i.d. data. This paper analyzes CV for correlated data. We present a criterion for suitability of standard CV in presence of correlations. When this criterion does not hold, we introduce a bias corrected cross-validation estimator which we term $CV_c,$ that yields an unbiased estimate of prediction error in many settings where standard CV is invalid. We also demonstrate our results numerically, and find that introducing our correction substantially improves both, model evaluation and model selection in simulations and real data studies. 
\end{abstract}

\noindent%
\vfill

\newpage
\spacingset{1.45}

\section{INTRODUCTION}
Datasets with correlation structures are common in modern statistical applications in various fields, such as geostatistics \citep{goovaerts1999geostatistics}, genetics \citep{maddison1990method} and ecology \citep{roberts2017cross}. Different modeling methods address the correlation structure differently. 
Some modeling methods, such as Gaussian process regression \cite[GPR]{rasmussen2006gaussian} and generalized least squares \citep[GLS]{hansen2007generalized}, utilize explicitly the correlation structure for achieving better prediction accuracy. Other predictive models, like random forest \citep[RF]{breiman2001random}, gradient boosting machines \citep[GBM]{friedman2002stochastic} and other machine learning models, do not consider explicitly the correlation structure but are still potentially able to utilize the correlation implicitly. 
The analysis in this paper mainly focuses on correlation that appears due to latent objects, such as random effects and random fields as appear in generalized linear mixed models \citep[GLMM]{verbeke1997linear} and generalized Gaussian process regression \citep[GGPR]{rasmussen2006gaussian} in clustered, temporal and spatial datasets. A simple example that demonstrates the way that latent variable realizations affect the correlation structure is linear mixed models \citep[LMM]{verbeke1997linear} with random intercept for clustered data:
\begin{gather} \label{LMM clustered}
y_{i,j}=\phibold_{i,j}^{t}\betabold+b_i+\epsilon_{i,j},\;i\in\{1,...,I\},\;j\in\{1,...,n_i\},
\end{gather}
where $y_{i,j}\in\mathbb{R}$ is the $j^{th}$ observation for cluster $i,$ $\phibold_{i,j}\in\mathbb{R}^{p}$ are the observed fixed effects covariates, $\betabold\in\mathbb{R}^{p}$ is the fixed effects coefficients vector, $b_i\sim N(0,\sigma_b^{2})$ are independent random effects and $\epsilon_{i,j}\sim N(0,\sigma_\epsilon^{2})$ are the i.i.d. errors. Since all the observations in cluster $i$ share the same random effect realization, they are correlated. For more information see \cite{verbeke1997linear}.

When it comes to prediction, the question of whether there is correlation between the observations from the \textit{training set} and the \textit{prediction set} --- the sample that is used for the model's parameter estimation, and the set of points whose response is predicted based on the trained model, respectively --- plays an important role. For example, in Eq. (\ref{LMM clustered}), if the training and prediction sets are sampled from the same clusters, then once the random effect realizations are estimated in the model training, they can be utilized for better prediction accuracy of the dependent variable in the prediction set. Under some conditions, a predictor that uses the estimated random effect realizations is the best linear unbiased predictor (BLUP), for more information see \cite{harville1976extension}. Another scenario is when the training and prediction sets are sampled from different clusters. In this scenario, observations of the prediction set are not correlated with observations of the training set, and therefore estimating the random effect realizations of the training set cannot be utilized for achieving better prediction accuracy. Of course, there are other correlation settings, e.g., when the random effect realizations of the training and the prediction sets are not the same but correlated. In addition to the distributional settings that are covered by GLMM and GGPR, explicit and implicit utilization of the correlation between the training and prediction sets is common in other distributional settings in various applications, including applications involving spatial datasets \citep{ward2018spatial}, longitudinal datasets \citep{hand2017practical} and datasets with hierarchical clustering structure \citep{raudenbush2002hierarchical}.

The correlation structure of the training and prediction sets, and in particular the correlation \textit{between} the training and prediction sets, can affect the model's prediction error and therefore should be carefully addressed when estimating the prediction error. Since many model selection procedures are based on prediction error estimation, ignoring the correlation setting may also affect model selection decision.

Various prediction error measures have been used and analyzed under different correlation settings, e.g., $AIC$-type \citep{vaida2005conditional} and $C_p$-type \citep{hodges2001counting}. This paper focuses on K-fold cross-validation prediction error estimator \cite[CV]{stone1974cross} which is the most widely used method for estimating prediction error \citep{Hastie2009elements}.
We introduce a new perspective on CV and propose an applicable framework for analyzing how CV is affected by the correlation structure in various distributional settings. In addition, a bias corrected cross-validation measure, $CV_{c},$ is introduced. $CV_{c}$ is suitable for many scenarios where CV is biased due to correlations.

Section \ref{section 2} presents the setting of the problem, the theoretical results of the proposed approach, and comparison with other methods. Section \ref{numerical part} presents numerical analyses that support the theoretical results of Section \ref{section 2}.

\section{CV FOR CORRELATED DATA}\label{section 2}
\subsection{K-fold Estimator and Generalization Error}\label{Setting}
Let $\{\xbold_{i}\}_{i=1}^{n},\;\xbold_{i}\in \mathbb{R}^p,$ be an independent and identically distributed (i.i.d.) sample form the probability distribution function $P_{\xbold}$. Let $\{y_i\}_{i=1}^{n},\;y_i\in \mathbb{R},$ be a sample that is drawn independently from the distribution $P_{y|\xbold=\xbold_{i},\sbold=\sbold_{0}},$ where $\sbold_{0}\in \mathbb{R}^q$ is the latent variable realization that induces correlation structure between $y_i\,'s.$

Denote $T=\{y_i,\xbold_{i}\}_{i=1}^{n}$ as the training sample for fitting a predictive model. Also denote $X=\left[\begin{array}{c}
-\;\boldsymbol{x}_{1}^{t}-\\
\vdots\\
-\;\boldsymbol{x}_{n}^{t}-
\end{array}\right]$ and $\ybold=\left[\begin{array}{c}
y_1\\
\vdots\\
y_n
\end{array}\right].$\\ 


One example of this setting is linear mixed models (LMM):
\[
\ybold=\Phi\betabold+Z\sbold+\epsilonbold,
\]
where $\Phi\in\mathbb{R}^{n\times (p-q)}$ contains the fixed effects covariates, $Z\in\mathbb{R}^{n\times q}$ contains the random effects covariates, $\sbold\in\mathbb{R}^{q}$ is a normally distributed random effects vector with a general covariance matrix and $\epsilonbold\in\mathbb{R}^n$ is the normally distributed error term with a general covariance matrix. In this setting $X=\{\Phi,Z\}.$ Eq. (\ref{LMM clustered}) is a special case of this setting. Extensions of this model are generalized linear mixed models \cite[GLMM]{breslow1993approximate}, where $\ybold=g^{-1}(\Phi\betabold+Z\sbold+\epsilonbold)$ and $g$ is the link function, as well as hierarchical generalized linear models \cite[HGLM]{lee1996hierarchical} where $\sbold$ and $\epsilonbold$ do not necessarily follow the normal distribution. Other examples of this setting that are commonly analyzed and represented using graphical probabilistic models tools are hidden Markov models and mixture models \citep{jordan2004graphical}.

Once a model is fitted, it is natural to evaluate the prediction ability of the model. In CV prediction error estimator, $T$ is randomly partitioned into K folds, $T_1=\{y_i,\xbold_{i}\}_{i=1}^{n_1},$ $T_2=\{y_i,\xbold_{i}\}_{i=(n_1+1)}^{n_2},...,T_K=\{y_i,\xbold_{i}\}_{i=(n_{K-1}+1)}^{n_K},$\footnote{We assume the observations are randomly ordered, so the folds are exchangeable.} where $n_k-n_{k-1}$ is the sample size in fold $k$ and $n_K=n.$ For each $k\in\{1,...,K\},$ the model is trained on the entire data except the $k^{th}$ fold, denoted by $T_{-k}=\bigcup_{j\neq k}T_j=\{\ybold_{-k},X_{-k}\}.$ The prediction error of the trained model is then measured on the holdout fold, $T_{k}=\{\boldsymbol{y}_{k},X_{k}\},$ with respect to some loss function $L(\cdot\,,\cdot):\;\mathbb{R}\times\mathbb{R}\to \mathbb{R}.$ Thus, each fold does not train the model used for predicting its outcome. The CV prediction error estimator is calculated  by averaging out the estimated prediction error across all the folds, i.e.,
\[
CV=\frac{1}{n}\sum_{k=1}^{K}\sum_{i\in k^{th}\,fold}L\big(y_i,\hat{y}(\xbold_{i};T_{-k})\big),
\]
where $\hat{y}(\xbold_{i};T_{-k})$ is the predictor of $y_{i},$ constructed by training with $T_{-K}$ and predicting on $\xbold_{i}.$ A special case of CV is leave one out CV \citep[LOO]{stone1974cross}, which is defined by setting $K=n,$ i.e., each observation defines a fold and therefore the model is trained on $n-1$ observations and its prediction error is evaluated on a single observation. Under some conditions, LOO is superior to other CV variants in terms of minimizing asymptotic bias and instability (for more information see \cite{burman1989comparative,arlot2010survey}), but can be computationally prohibitive in some settings. A comprehensive comparison of CV variants in terms of prediction error estimation and models selection can be found in \cite{zhang2015cross}.

In case of i.i.d. sampling, CV is considered to be an estimator of the generalization error, which is defined here in a wide sense that also covers settings with correlated data:
\begin{align}\label{generalization error}
generalization\; error=\E_{T_{te},T_{tr}}\frac{1}{n_{te}}\sum_{i=1}^{n_{te}}L\big(y_{te,i},\Hat{y}(\xbold_{te,i};T_{tr})\big),
\end{align}
where 
\begin{itemize}
    \item $T_{tr}=\{\ybold_{tr},X_{tr}\}=\{y_{tr,i},\xbold_{tr,i}\}_{i=1}^{n_{tr}}$
    \begin{itemize}
        \item $\xbold_{tr,i}$ is an i.i.d. sample from $P_{\xbold}$
        \item $y_{tr,i}$ is sampled from $P_{y|\xbold=\boldsymbol{x}_{tr,i},\sbold=\sbold_{0}}$
    \end{itemize}
    \item $T_{te}=\{\boldsymbol{y}_{te},X_{te}\}=\{y_{te,i},\xbold_{te,i}\}_{i=1}^{n_{te}}$
    \begin{itemize}
        \item $\xbold_{te,i}$ is an i.i.d. sample from $P_{\xbold}$
        \item $y_{te,i}$ is sampled from $P_{y|\xbold=\xbold_{te,i},\sbold=\sbold_{te}}$
    \end{itemize}
   \end{itemize}
The reason that CV is commonly considered to be an estimator of the generalization error is the random mechanism that is embedded in the CV procedure.
In this perspective, $T_k$ and $T_{-k}$ are equivalent to $T_{te}$ and $T_{tr},$ respectively. Then, averaging the prediction error, $\sum_{i\in k^{th}\,fold} L\big(y_i,\hat{y}(\xbold_{i};T_{-k})\big)/\big(n_k-n_{k-1}\big),$ over the different folds, estimates Eq. (\ref{generalization error}). 

It is important to emphasize that unlike $T,$ which is the available dataset for modeling, $\{T_{tr},T_{te}\}$ are used in the paper for demonstrating different prediction goals, rather than actual datasets.

\begin{remark}
Note, it is typically assumed that $T_{tr}$ is distributed as $T$ and therefore of size $n.$ In this case, the size of $T_{-k}$ and $T_{tr}$ is obviously different, and additional bias in CV evaluation is introduced. This is typically ignored, especially when considering LOO and assuming training with $n-1$ or $n$ observations carries little difference. In what follows we also ignore this and implicitly assume that $T_{tr}$ is of the same size as $T_{-k}$ for all $k$.
\end{remark}

When
\begin{align} \label{ind CV}
\{\ybold_{k},X_{k}\}\independent \{\ybold_{-k},X_{-k}\}, \;\forall k\in \{1,...,K\}
\end{align}
and
\[
\{\boldsymbol{y}_{te},X_{te}\}\independent \{\ybold_{tr},X_{tr}\},
\]
it is clear that 
CV is an unbiased estimator of the generalization error, however a careful analysis is required for the case when the folds are dependent.

Next we investigate how deviating from the condition in Eq. (\ref{ind CV}) contributes a bias to CV with respect to the generalization error. Based on it, a bias corrected CV estimator, $CV_c,$ will be presented. 

In what follows we limit the discussion to squared error loss function. Also note, from now on LOO setting will be assumed, and consequently $n_{te}=1, n_{tr}=n-1,$ however the results are valid to other CV partitioning settings. In particular, it can easily be seen that the size of the test set, $n_{te},$ has no bearing on generalization error definition.



\subsection{A General Formulation of CV Bias}
Let
\begin{align*}
w_{cv}=\E_{T_{te},T_{tr}}\big(y_{te}-\Hat{y}(\xbold_{te};T_{tr})\big)^2-\frac{1}{n}\E_{T}\|\ybold-\Hat{\ybold}_{cv}(T)\|^2,
\end{align*}
where $\Hat{\ybold}_{cv}(T)=\big(\hat{y}(\xbold_1;T_{-1}),...,\hat{y}(\xbold_n;T_{-n})\big)^{t}$ is the CV predictor of $\ybold.$

An unbiased estimator of the generalization error is:
\begin{align}\label{CV Optimism}
CV_c=CV+w_{cv}.
\end{align}
Before analyzing different correlation settings, we derive a more explicit expression for $w_{cv}.$ For simplicity, subscript notations in the expectation operator are frequently omitted. Therefore, unless a specific object is specified, $\E$ averages all the random variables that it operates on.

Since
\begin{align*}
    \E\big(y_{te}-\Hat{y}(\xbold_{te};T_{tr})\big)^2=&\E\big(y_{te}-\E\,y_{te}\big)^2+\big(\E\,y_{te}-\E\,\Hat{y}(\xbold_{te};T_{tr})\big)^2\\
    &+\E\big(\E\,\Hat{y}(\xbold_{te};T_{tr})-\Hat{y}(\xbold_{te};T_{tr})\big)^2\\
    &+2\E\big(y_{te}-\E\,y_{te}\big)\big( \E\,\Hat{y}(\xbold_{te};T_{tr})-\Hat{y}(\xbold_{te};T_{tr})\big),\\
     \\
     \E\|\ybold-\Hat{\ybold}_{cv}(T)\|^2=&\E\|\ybold-\E\,\ybold\|^2+\|\E\,\ybold-\E\,\Hat{\ybold}_{cv}(T)\|^2+\E\|\E\,\Hat{\ybold}_{cv}(T)-\Hat{\ybold}_{cv}(T)\|^2\\
    &+2\E\big(\ybold-\E\,\ybold\big)^{t}\big(\E\,\Hat{\ybold}_{cv}(T)-\Hat{\ybold}_{cv}(T)\big),
\end{align*}
and $y_i$ and $y_{te}$ have the same marginal distribution, which gives $\E\|\ybold-\E\,\ybold\|^2/n-\E\big(y_{te}-\E\,y_{te}\big)^2=0,$ then:
\begin{align}\label{w_{cv} general derivation}
w_{cv}=&\big(\E\,\Hat{y}(\xbold_{te};T_{tr})-\E\,y_{te}\big)^2-\frac{1}{n}\|\E\,\Hat{\ybold}_{cv}(T)-\E\,\ybold\|^2\\ \nonumber
&\E\big(\Hat{y}(\xbold_{te};T_{tr})-\E\,\Hat{y}(\xbold_{te};T_{tr})\big)^2-\frac{1}{n}\E\|\Hat{\ybold}_{cv}(T)-\E\,\Hat{\ybold}_{cv}(T)\|^2\\ \nonumber
&-2 \E\,\Cov\big(\Hat{y}(\xbold_{te};T_{tr}),y_{te}\big)+\frac{2}{n}\E\,\tr\big[\Cov\big(\Hat{\ybold}_{cv}(T),\ybold\big)\big],
\end{align}
where $\tr$ is the trace operator and $\Cov$ is the covariance operator which contains conditional expectation of the dependent variables $\ybold,\;y_{te}$ and $\ybold_{tr}$ given their covariates, $X,\;\xbold_{te}$ and $X_{tr}$ e.g., $\E\,\tr\big[\Cov\big(\Hat{\ybold}_{cv}(T),\ybold\big)\big]\coloneqq\E_{X}\tr\big[\Cov\big(\Hat{\ybold}_{cv}(T),\ybold|X\big)\big].$


The first two lines in Eq. (\ref{w_{cv} general derivation}) are the differences between the bias and the variance of $\Hat{\ybold}_{cv}(T)$ and $\Hat{y}(\xbold_{te};T_{tr}),$ where the expectation is taken also over the covariates. The third line relates to the covariances between the response and its predictor in each scheme --- CV prediction error and generalization error.


\subsection{Criterion for CV Unbiasedness}\label{section same s}
Let $P_{T_{te},T_{tr}}$ and $P_{T_{k},T_{-k}}$ be the joint distributions of $\{T_{te},T_{tr}\}$ and $\{T_{k},T_{-k}\},$ respectively. Theorem \ref{wcv for same s} describes a simple generic condition when no correction is required for CV.

\begin{theorem}\label{wcv for same s}
If $P_{T_{te},T_{tr}}=P_{T_{k},T_{-k}}\;\forall
k\in \{1,..,n\},$ then $w_{cv}=0.$
\end{theorem}
\begin{proof}$\;$
Since $\{T_{k},T_{-k}\}$ were drawn from the same distribution as $\{T_{te},T_{tr}\},$ then an expectation over any transformation of them is equal, in particular:
\[
 \big(\E\, y_{te}-\E\,\Hat{y}(\xbold_{te};T_{tr})\big)^2=\frac{1}{n}\|\E\,\ybold-\E\,\Hat{\ybold}_{cv}(T)\|^2.
\]
Similarly
\begin{align*}
\E\big(\Hat{y}(\xbold_{te};T_{tr})-\E\,\Hat{y}(\xbold_{te};T_{tr})\big)^2&=\frac{1}{n}\E\|\Hat{\ybold}_{cv}(T)-\E\,\Hat{\ybold}_{cv}(T)\|^2\\ 2 \E\,\Cov\big(\Hat{y}(\xbold_{te};T_{tr}),y_{te}\big)&=\frac{2}{n}\E\,\tr\big[\Cov\big(\Hat{\ybold}_{cv}(T),\ybold\big)\big].  
\end{align*}
\end{proof}
Theorem \ref{wcv for same s} states a very basic and intuitive condition of CV unbiasedness --- When the CV partitioning preserves the distributional relation between the prediction set to the training set, then CV is unbiased. 

Let $P_{y_k|\xbold_{k},T_{-k}},\;P_{y_{te}|\xbold_{te},T_{tr}}$ and $P_{y_{te}|\xbold_{te},T_{-k}}$ denote the respective conditional distribution. The condition in Theorem \ref{wcv for same s} can be compacted to the following one:
\[
P_{y_k|\xbold_{k},T_{-k}}=P_{y_{te}|\xbold_{te},T_{tr}}.
\]
Moreover, since $T_{tr}$ is assumed to be distributed as $T_{-k},$ then this condition can be rewritten as follows:
\begin{align}\label{condition}
P_{y_k|\xbold_{k},T_{-k}}=P_{y_{te}|\xbold_{te},T_{-k}}.
\end{align}

Of course, when $w_{cv}=0$ the CV is unbiased and therefore suitable. Based on Theorem \ref{wcv for same s}, it is important to stress that CV is suitable not only for the case that $\{y_k,\xbold_{k}\}_{k=1}^{n}$ are independent, neither only for the case when $\{y_k,\xbold_{k}\}_{k=1}^{n}$ are exchangeable, as is commonly mistakenly referred in the literature \citep{anderson2018comparing,roberts2017cross}. 
The biasedness of CV only relates to the question whether $T_{-k}$ contributes more information for predicting $y_{k}$ than $T_{tr}$ contributes for predicting $y_{te}.$

We can demonstrate the use of Theorem \ref{wcv for same s} for a simple application --- using LMM for predicting new observations from the \textit{same} clusters that appear in the training set. In the case $\sbold_{te}=\sbold_{0},$ i.e., \begin{align}\label{LMM same s}
\ybold=&\Phi\betabold+Z\sbold_{0}+\epsilonbold\\\nonumber
\ybold_{tr}=&\Phi_{tr}\betabold+Z_{tr}\sbold_{0}+\epsilonbold_{tr}\\\nonumber
y_{te}=&\phibold_{te}\betabold+\zbold_{te}\sbold_{0}+\epsilon_{te},
\end{align}
and
\begin{align*}
X=&\{\Phi,Z\},\;\Phi\in\mathbb{R}^{n\times (p-q)},\;Z\in\mathbb{R}^{n\times q}\\
X_{tr}=&\{\Phi_{tr},Z_{tr}\},\;\Phi_{tr}\in\mathbb{R}^{(n-1)\times (p-q)},\;Z_{tr}\in\mathbb{R}^{(n-1)\times q}\\ \xbold_{te}=&\{\phibold_{te},\zbold_{te}\},\;\phibold_{te}\in\mathbb{R}^{1\times (p-q)},\;\zbold_{te}\in\mathbb{R}^{1\times q},\\
\end{align*}
$\sbold_{0}\in\mathbb{R}^{q}$ is the random effect realization vector, where each entry is a random effect realization for a different cluster and $\epsilonbold\in\mathbb{R}^{n},\;\epsilonbold_{tr}\in\mathbb{R}^{(n-1)},\;\epsilon_{te}\in\mathbb{R}$ are i.i.d. normal errors terms.

As was mentioned previously, the observations in $X,\;X_{tr}$ and $\xbold_{te}$ are i.i.d.. Also, in this example $\ybold,\;\ybold_{tr}$ and $y_{te}$ were drawn given the same latent variable realization, $\sbold_{0},$ therefore Theorem \ref{wcv for same s}'s condition --- $P_{y_k|\xbold_{k},T_{-k}}=P_{y_{te}|\xbold_{te},T_{-k}}$ --- is satisfied and $w_{cv}=0.$ This use case of predicting new points from the same clusters that were used in the training data is common, for example see \cite{gelman2006multilevel}.


The principle of CV suitability for the setting in Eq. (\ref{LMM same s}) is discussed in the LMM literature \citep{fang2011asymptotic,little2017using}, however we did not find any general mathematical formalization of it. Commonly, CV is avoided in applications involving correlated data based on the wrong perception that CV is always unsuitable for these cases \citep{anderson2018comparing,roberts2017cross}. It is also important to stress that since the condition in Theorem \ref{wcv for same s}  only relates to the distributional relation between $\{T_{te},T_{tr}\}$ and $\{T_{k},T_{-k}\}$ rather than specifying a distribution, Theorem \ref{wcv for same s} can be implemented in applications where the distributional settings are not fully specified, as is common when implementing machine learning algorithms.

\subsection{CV Correction}\label{section new s}
Now, consider the setting where Theorem \ref{wcv for same s}'s condition is not satisfied, i.e., when $P_{T_{te},T_{tr}}\neq P_{T_{k},T_{-k}}.$ A simple scenario for this setting can again be taken from LMM for clustered data, where $T_{te}$ contains a \textit{new} latent random effects realization i.e.,
\begin{align} \label{LMM for new b}
\ybold=&\Phi\betabold+Z\sbold_{0}+\epsilonbold\\\nonumber
\ybold_{tr}=&\Phi_{tr}\betabold+Z_{tr}\sbold_{0}+\epsilonbold_{tr}\\ \nonumber
y_{te}=&\phibold_{te}\betabold+\zbold_{te}\sbold_{te}+\epsilon_{te},
\end{align}
where $\sbold_{te}\neq \sbold_{0}.$ Other relevant components are defined in the same way as in Eq. (\ref{LMM same s}). In this scenario the correlation between $\ybold_{tr}$ and $y_{te}$ is different than the correlation between $\ybold_{-k}$ and $y_{k},$ and a further analysis of $w_{cv}$ is required. This can occur for example when the random effects are intercepts for clusters (e.g., cities), and the data was collected at one point of time, but the prediction task is performed on the same clusters at another point of time.


Here, first we will find an estimator of $w_{cv}$ when $\Hat{\ybold}_{cv}(T)$ is linear in $\ybold,$ then $w_{cv}$ for nonlinear predictors will be discussed as well. However, before that, let us demonstrate how $\Hat{\ybold}_{cv}(T)$ can be formalized linearly in $\ybold$ for some models.
\begin{definition}
$\Hat{\ybold}_{cv}(T)$ is linear in $\ybold$ if:
\begin{align}\label{Linear CV}
\Hat{\ybold}_{cv}(T)=H_{cv}\ybold,
\end{align}
where $H_{cv}\in\mathbb{R}^{n\times n}$ does not contain $\ybold$ and is constructed as follows:
\[
H_{cv}=\left[\begin{array}{cccc}
0 & h_{1,2} & ... & h_{1,n}\\
h_{2,1} & 0 &  & h_{2,n}\\
...\\
h_{n,1} & h_{n,2} & ... & 0
\end{array}\right],
\]
$h_{k,k'}\in\mathbb{R}\;\forall k,k'\in \{1,...,n\}.$ 
\end{definition}
The cross validation's principle that $y_{k}$ is not involved in predicting itself is reflected by having zeros in the diagonal, therefore $\hat{y}(\xbold_{k},T_{-k})=\hbold_k \ybold_{-k},$ where $\hbold_k\in\mathbb{R}^{1\times n-1}$ is the $k^{th}$ row of $H_{cv}$ without the diagonal element, e.g.,
\[
\hbold_1=[h_{1,2},...,h_{1,n}].
\]
Of course, $H_{cv}$ can be defined for other K-fold CV settings by adjusting the dimension of the blocks with the zeros on the diagonal (in LOO the dimension of each block is $1\times 1,$ however for general K-fold CV is $n/K\times n/K$).

Examples for linear models are ordinary least squares (OLS), generalized least squares (GLS), ridge regression, smoothing splines, LMM, Gaussian process regression (GPR) and kernel regression. 
\begin{definition}\label{htest}
Let $\hbold_{te}\in\mathbb{R}^{1\times n-1}$ be the hat vector of $y_{te},$ constructed from $\{X_{tr},\Cov(y_{te},\ybold_{tr})\}$ in the same way as $\hbold_{k}$ is constructed from $\{X_{-k},\Cov(y_{k},\ybold_{-k})\}.$
\end{definition}
By definition \ref{htest}, $\hbold_{te}$ is the equivalent of $\hbold_{k}$ for the set $\{T_{te},T_{tr}\}$ and they are distributed the same.
\begin{theorem}\label{EDF Theorem}
Let $\Hat{\ybold}_{cv}(T)=H_{cv}\ybold$ be a linear predictor of $\ybold,$ and $\Hat{y}_{te}(\xbold_{te},T_{tr})=\hbold_{te}\ybold_{tr}$ is its corresponding predictor of $y_{te}.$ Then:
\[
w_{cv}=\frac{2}{n}\E\big[\tr\Big(H_{cv}\Cov\big(\ybold,\ybold\big)\Big)-n\hbold_{te}\Cov\big(\ybold_{tr},y_{te}\big)\big].
\]
\end{theorem}


\begin{proof}
Based on Eq. (\ref{w_{cv} general derivation}) and by assuming linear predictor, 
\[
-2 \E\,\Cov\big(\Hat{y}(\xbold_{te};T_{tr}),y_{te}\big)+\frac{2}{n}\E\,\tr\big[\Cov\big(\Hat{\ybold}_{cv}(T),\ybold\big)\big]=w_{cv}.
\]
Therefore, it is only left to show that
     \begin{align}\label{equal bias}
          \big(\E\,y_{te}-\E\,\hbold_{te}\ybold_{tr}\big)^2&=\frac{1}{n}\|\E\,\ybold-\E\,H_{cv}\ybold\|^2\\  \label{equal variance}
\E\big(\hbold_{te}\ybold_{tr}-\E\,\hbold_{te}\ybold_{tr}\big)^2&=\frac{1}{n}\E\| H_{cv}\ybold-\E\,H_{cv}\ybold\|^2. 
     \end{align}
     Since, by definition $y_{k}$ and $y_{te}$ were drawn from the same distortion, and
$\ybold_{-k}$ and $\ybold_{tr}$ were drawn from the same distribution, then although  $P_{T_{te},T_{tr}}\neq P_{T_{k},T_{-k}},$ still  
     \[
      \|\E\,\ybold-\E\,H_{cv}\ybold\|^2=n\big(\E\,y_{k}-\E\,\hbold_{k}\ybold_{-k}\big)^2=n\big(\E\,y_{te}-\E\,\hbold_{te}\ybold_{tr}\big)^2.
     \]
     Therefore Eq. (\ref{equal bias}) holds. Similarly with Eq. (\ref{equal variance}).
\end{proof}

Using Theorem \ref{EDF Theorem}, given $\Cov\big(\ybold,\ybold\big)$ and $\Cov\big(\ybold_{tr},y_{te}\big),$ an estimator of the generalization error for a linear predictor is:
\begin{align}\label{Prediction error estimator}
\widehat{CV_c}=\frac{1}{n}\big(\ybold-H_{cv}\ybold\big)^{t}\big(\ybold-H_{cv}\ybold\big)+\frac{2}{n}\Big[\tr\Big(H_{cv}\Cov\big(\ybold,\ybold\big)\Big)-n\hbold_{te}\Cov\big(\ybold_{tr},y_{te}\big)\Big].
\end{align}

A special case is when $\Cov\big(\ybold_{tr},y_{te}\big)=0.$ For example, in the clustered LMM setting that was given in Eq. (\ref{LMM for new b}), $\Cov\big(\ybold_{tr},y_{te}\big)=0$ when the latent variable realizations of $\ybold_{tr}$ and $y_{te},$ i.e., $\sbold_{0}$ and $\sbold_{te},$ are independent. In this case:
\begin{align}\label{Prediction error estimator for independent s}
\widehat{CV_c}=\frac{1}{n}\big(\ybold-H_{cv}\ybold\big)^{t}\big(\ybold-H_{cv}\ybold\big)+
\frac{2}{n}\tr\big[H_{cv}\Cov\big(\ybold,\ybold\big)\big].
\end{align}

In case $\Hat{\ybold}_{cv}(T)$ is not a linear predictor and Theorem \ref{wcv for same s}'s condition does not hold, then there is no closed form for $w_{cv}.$ However, a correction is required as CV is not an unbiased estimator of the prediction error in this case. See Section \ref{Comparison with other methods} for some ad-hoc solutions that may apply in specific cases. 

It is important to note that:
\begin{itemize}
    \item Theorem \ref{EDF Theorem}, as well as Eqs. (\ref{Prediction error estimator}) and (\ref{Prediction error estimator for independent s}), are valid for any K-fold setting and not only for LOO. Of course, $H_{cv}$ is based on $K$ and should be adjusted correspondingly.
\item The correction term in $\widehat{CV_c}$ depends on covariance matrices, which are commonly estimated. The effect of using estimated covariance matrices instead of the true covariance matrices is discussed and analyzed in Sections \ref{Estimating covariance} and \ref{numerical part}.
\end{itemize}

In case there are several alternative models, it is common to use their estimated prediction errors for selecting the best model. 
\begin{definition}
Given a set of models $\mathcal{H},$ the best $\widehat{CV_c}$ model is:
\[
\underset{h\in\mathcal{H}}{\argmin}\;\widehat{CV_c}(h),
\]
where $\widehat{CV_c}(h)$ is $\widehat{CV_c}$ for model $h.$
\end{definition}

\subsubsection{Specifying \texorpdfstring{$\hbold_{te}$}{TEXT} and \texorpdfstring{$\Cov\big(\ybold_{tr},y_{te}\big)$}{TEXT}}\label{Specifying the correction}
The term $n\hbold_{te}\Cov\big(\ybold_{tr},y_{te}\big)$ in Eq. (\ref{Prediction error estimator}) is relevant when $\ybold_{tr}$ and 
$y_{te}$ are correlated (but not in the same way as $y_{k}$ and $\ybold_{-k})$ e.g., assuming $T=\{X,\ybold\}$ has a hierarchical clustering structure: 
\begin{gather*}
    y_{i,j,r}=\xbold_{i,j,r}^{t}\betabold+u_{i}+b_{i,j}+\epsilon_{i,j,r},\\
    u_{i}\sim N(0,\sigma^2_{u}),\;b_{i,j}\sim N(0,\sigma^2_{b}),\ \epsilon_{i,j,r}\sim N(0,\sigma^2_{\epsilon}),\\
    \forall\; i\in\{1,...,I\},\;j\in\{1,...,J\},r\in\{1,...,R\}.
\end{gather*}
Assume that the prediction goal is to estimate:
\[
y_{te}=\xbold_{te}^{t}\betabold+u_{te}+b_{te}+\epsilon_{te},
\]
where $u_{te}\in [u_1,...,u_I],$ however $b_{te}$ does not depend on the realizations of $\{b_{i,j}\}_{i\in\{1,...I\},\;j\in\{1,...,J\}}.$ In words, the prediction goal is to predict a new observation that relates to one of the high-level clusters (indexed by $i$), but from a new low-level cluster (indexed by $j$).

In this setting, since the realization of $u_{te}$ appears in $\ybold,$ then $\Cov(\ybold_{tr},y_{te})\neq0$ and $n\hbold_{te}\Cov\big(\ybold_{tr},y_{te}\big)$ is required for calculating $\widehat{CV_c}.$

Although $T_{te}$ is an abstract object, $\hbold_{te}$ and $\Cov(\ybold_{tr},y_{te})$ can be extracted from the data. Since $\hbold_{te}$ is distributed as 
$\hbold_{k}$ (for random k), then $\hbold_{te}$ can be extracted by selecting a random $\hbold_{k}$ from $H_{cv}.$ $\Cov(\ybold_{tr},y_{te})$ is extracted correspondingly as $\Cov(\ybold_{-k},y_{k}|\boldsymbol{b},\epsilonbold),$ where $\boldsymbol{b}$ and $\epsilonbold$ are vectors containing $\{b_{i,j}\}_{i\in\{1,...I\},\;j\in\{1,...,J\}}$ and $\{\epsilon_{i,j,r}\}_{i\in\{1,...I\},\;j\in\{1,...,J\},\;r\in\{1,...,R\}},$ respectively. By conditioning on these vectors, the remaining covariance stems from $\boldsymbol{u}=[u_1,...,u_I]$ only, corresponding to the covariance between $\ybold_{tr}$ and $y_{te}.$

Simpler, $n\hbold_{te}\Cov\big(\ybold_{tr},y_{te}\big)$ can be replaced by 
$\tr\big(H_{cv}\Cov(\ybold,\ybold|\boldsymbol{b},\epsilonbold)\big),$ such that:
\begin{align*}
\widehat{CV_c}=&\frac{1}{n}\big(\ybold-H_{cv}\ybold\big)^{t}\big(\ybold-H_{cv}\ybold\big)+\frac{2}{n}\Big[\tr\Big(H_{cv}\Cov\big(\ybold,\ybold\big)\Big)-\tr\Big(H_{cv}\Cov\big(\ybold,\ybold|\boldsymbol{b},\epsilonbold\big)\Big)\Big]\\
=&\frac{1}{n}\big(\ybold-H_{cv}\ybold\big)^{t}\big(\ybold-H_{cv}\ybold\big)+\frac{2}{n}\tr\Big(H_{cv}\Cov\big(\ybold,\ybold\big|\boldsymbol{u})\Big).
\end{align*}

\subsubsection{Interpretation of the Results}\label{Interpretation of the Results}
The correction $2/n\times\Big[\tr\Big(H_{cv}\Cov\big(\ybold,\ybold\big)\Big)-n\hbold_{te}\Cov\big(\ybold_{tr},y_{te}\big)\Big]$
is intuitive since it expresses the difference between the correlation structure of the target prediction problem and the correlation structure in the available dataset, $T.$ 

Theorems \ref{wcv for same s} and \ref{EDF Theorem} emphasize that the question whether CV is biased relates to the distributional setting and it is indifferent to the implemented algorithm. The implemented algorithm is expressed only by the values of $H_{cv}$ in $\widehat{CV_c}.$ 

An interesting example that stresses this understanding is when generalized least squares (GLS) is implemented in a use case with a correlation setting of $\Cov\big(\ybold_{tr},y_{te}\big)=0.$ Since GLS estimates $\E(\ybold|X)$ and does not utilize explicitly the random effects for achieving better prediction accuracy --- as LMM does by estimating $\E(\ybold|X,\sbold)$ --- then one may think that CV is unbiased when GLS is implemented, regardless of the correlation structure of $\Cov\big(\ybold_{tr},y_{te}\big).$ However, as was mentioned, this is wrong since the CV biasedness relates to the distributional setting rather than to the implemented algorithm. The bias in this case is:
\[
w_{cv}(GLS)=\E\,\sum_{k=1}^{n}\xbold_{k}^{t}\Big(X_{-k}^{t}\Cov(\ybold_{-k},\ybold_{-k}\big)^{-1}X_{-k}\Big)^{-1}X_{-k}^{t}\Cov\big(\ybold_{-k},\ybold_{-k}\big)^{-1}\Cov\big(\ybold_{-k},y_{k}\big).
\]
If $\Cov(\ybold,\ybold)=\sigma^2_{\epsilon}I+\rho\mathbbm{1}_{n\times n},$ where $\mathbbm{1}_{n\times n}$ is a $n$ by $n$ matrix of ones and $\rho\in\mathbb{R}^+,$ then:
\begin{align*}
w_{cv}&=\frac{\rho}{\sigma^2_{\epsilon}+\rho\big(n-1\big)}\sum_{k=1}^{n}\E\,\xbold_{k}\Big(X_{-k}^{t}\Cov\big(\ybold_{-k},\ybold_{-k}\big)^{-1}X_{-k}\Big)^{-1}X_{-k}^{t}\mathbbm{1}_{n-1},
\end{align*}
where the identity $\Cov\big(\ybold_{-k},\ybold_{-k}\big)^{-1}\Cov\big(\ybold_{-k},y_{k}\big)=\mathbbm{1}_{n-1}\rho/\Big(\sigma^2_{\epsilon}+\rho\big(n-1\big)\Big)$ is based on the result by \cite{miller1981inverse}.

The intuition behind the biasedness of CV in this setting is that although GLS does not utilizes explicitly estimated random effect realizations, its estimated model coefficients are still affected by the random effect realizations in $T.$ Similarly, if GLS was replaced by OLS,  CV would still be biased under this correlation setting.

It is also important to emphasize that Theorem \ref{wcv for same s} derives $w_{cv}$ explicitly ($w_{cv}=0$) for any model under its assumed conditions, however $\widehat{CV_c}$ uses an approximated $w_{cv},$ which is only relevant for linear models.

\subsubsection{Comparison with Expected Optimism}
Below, a comparison between the correction in $\widehat{CV_c}$ and the expected optimism correction \citep{efron1986biased} is presented.

Expected optimism correction was developed in a context of in-sample prediction error measure:
\[
in{-}sample\;error=\E_{\boldsymbol{y}^*,\ybold}\frac{1}{n}\|\boldsymbol{y}^{*}-\hat{\ybold}(X;T)\|^2,
\]
where $\boldsymbol{y}^*\in\mathbb{R}^n$ is identically distributed but independent copy of $\ybold.$ The in-sample prediction error is estimated by
\[
\frac{1}{n}\|\ybold-\hat{\ybold}(X;T)\|^2+w,
\]
where $w$ is the expected optimism:
\[
w=\E_{\ybold^*,\ybold}\big(\frac{1}{n}\|\boldsymbol{y}^*-\hat{\ybold}(X;T)\|^2-\frac{1}{n}\|\ybold-\hat{\ybold}(X;T)\|^2\big).
\]
If $\hat{\ybold}(X;T)=H\ybold,$ for some hat matrix $H,$ then
\[
w=\frac{2}{n}\tr\Big(H\Cov\big(\ybold,\ybold\big)\Big).
\]

The similarity between $w$ and the correction in $\widehat{CV_c}$ in case $\Cov\big(\ybold_{tr},y_{te}\big)=0,$ i.e., in case $w_{cv}=\E\Big(2/n\times\tr\big(H_{cv}\Cov(\ybold,\ybold)\big)\Big),$ is interesting since it reflects the relation between generalization error and in-sample error and emphasizes the role of the linearity in this relation.

The fundamental difference between in-sample error and generalization error is that in the latter, the covariates matrices, $X_{tr},\; \xbold_{te}$ are assumed to be random variables and therefore in generalization error, unlike in in-sample prediction error:
\begin{enumerate}
    \item $X_{tr}$ and $\xbold_{te}$ are not identical.
    \item An expectation is taken also over $\{X_{tr},\;\xbold_{te}\}.$
\end{enumerate}
As was mentioned in the previous sections, the inner-sampling mechanism of $\{T_{k},T_{-k}\}$ in CV that emulates repeated sampling of $\{T_{te},T_{tr}\}$ from $\{P_{\xbold}, P_{y|\xbold,\sbold}\}$ expresses these properties. These properties are also reflected in the correction. Since
\[
\frac{2}{n}\tr\Big(H_{cv}\Cov\big(\ybold,\ybold\big)\Big)=\frac{2}{n}\sum_{k=1}^{n}\hbold_{k}\Cov(y_{k},\ybold_{-k}),
\]
then $2/n\times\tr\Big(H_{cv}\Cov\big(\ybold,\ybold\big)\Big)$ averages $n$ identically distributed atoms, $\hbold_{k}\Cov(y_{k},\ybold_{-k}),$ where each one of them is an unbiased estimator of $w_{cv}.$ Unlike in $w,$ which relates to a specific covariates matrix realization, $X,$ the atoms in $2/n\times\tr\Big(H_{cv}\Cov\big(\ybold,\ybold\big)\Big)$ relate to different covariates realizations, $\{X_{-k},\xbold_{k}\}_{k=1}^{n}.$

In addition, when $\Cov(\ybold,\ybold)=\sigma^2_{\epsilon}I,$ while $w_{cv}=0,$ $w=2\sigma^2_{\epsilon}/n\times\tr(H).$ This is since in CV the sample is partitioned into training and test, however in expected optimism approach the whole sample is used for both tasks --- training and test.

\subsection{Advanced Correlation Settings}

\subsubsection{Kriging}\label{Kriging}
Many applications with spatial and temporal data are analyzed using \textit{random functions} framework, rather than multivariate random variable framework.
A comprehensive review about random functions data analysis can be found in \cite{wang2016functional}. Here we focus on interpreting the results from Sections \ref{section same s} and \ref{section new s} in the context of a specific use case in the random functions framework --- Kriging with Gaussian process regression \citep[GPR]{rasmussen2006gaussian}. Numerical analysis of $\widehat{CV_c}$ implementation in real spatial dataset is presented in Section \ref{California housing, example}.

In Kriging \citep{goovaerts1999geostatistics} the goal is to create a climate map on some surface, $\mathbb{A},$ using climate predictions at a high-resolution grid of the surface. The predictions at the grid points are based on a predictive model that was fitted to a sample, $T,$ that was drawn from this surface, but covers the surface sparsely. In many cases, GPR is the predictive modeling method that is used for Kriging. In this method, as well as in other functional data analysis methods, the mean and the covariance of the predicted variable are formulated as functions. The mean function typically depends on fixed effects of some covariates (like elevation). The estimated mean function in GPR is a linear function of $\ybold.$ The covariance function, which is termed the \textit{kernel function}, $\mathcal{K_{\mathbb{A}}}:(\mathbb{A}\times \mathbb{A})\to \mathbb{R},$ measures the covariance between each two points in $\mathbb{A},$ whether the points are in the sample $T$ or not. Unlike in the multivariate approach, where the correlation is induced by a latent random variable realization, in the functional approach the correlation structure of the surface $\mathbb{A},$ as it is expressed by $\mathcal{K}_{\mathbb{A}},$ is induced by realization of a stochastic process instance --- \textit{latent random function}, $\sbold.$
Since Theorems \ref{wcv for same s} and \ref{EDF Theorem} are based on the relation between $P_{y_{k}|\xbold_{k},T_{-k}}$ and $P_{y_{te}|\xbold_{te},T_{-k}},$ rather than whether the source of the correlation between the observations is a latent random variable or latent random function, these theorems can also be applied here.

Let us consider three scenarios. The first scenario is the classical Kriging use case, where the observations of both samples, $T_{te}$ and $T,$ are randomly sampled from the same surface, $\mathbb{A},$ i.e., observations of both samples are drawn independently from $\{P_{\xbold},\;P_{y|\xbold,\sbold=\sbold_{0}}\},$ where $\sbold_{0}$ is the realization of the latent random function $\sbold$ in the surface $\mathbb{A}.$ In this case Theorem \ref{wcv for same s}'s condition is satisfied and therefore $w_{cv}=0$ and CV is suitable.

The second scenario is when the realization of $\sbold$ is not the same in $T_{te}$ and  $T,$ and therefore while the observations of $T$ follow $\{P_{\xbold},\;P_{y|\xbold,\sbold=\sbold_{0}}\},$ the observation in $T_{te}$ follows $\{P_{\xbold},\;P_{y|\xbold,\sbold=\sbold_{te}}\}.$ In this case Theorem \ref{wcv for same s}'s condition is not satisfied. An example for this scenario is when $T_{te}$ is sampled from the same surface as $T,$ $\mathbb{A},$ however at a future time-point (e.g., when the goal it to create a climate map for the next year based on this year's data). In this case, $y_k$ and $\ybold_{-k},$ which are sampled at the same time point, are more correlated than $y_{te}$ and $\ybold_{-k},$ which are sampled at different time-points. Therefore, CV is biased. As we saw in Section \ref{section new s}, if a linear model (such as GPR) is used, then $\widehat{CV_c}$ is an unbiased estimator of the generalization error and therefore should be used instead of CV.

Another spatial application in this scenario is when $T_{te}$ is sampled from the surface $\mathbb{A}',$ which is different than  $\mathbb{A}.$ Since the surfaces are different, then their latent random function realizations are different. Therefore, assuming observations in both samples, $T$ and $T_{te},$ were drawn from the \textit{same marginal distribution} --- $\{P_{\xbold},\;P_{y|\xbold}\}$ ---, then $\widehat{CV_c}$ should be used instead of CV.

Another interesting scenario that is not covered either by Theorem \ref{wcv for same s} or by Theorem \ref{EDF Theorem} is when $\{y_k,\xbold_{k}\}_{k=1}^{n}$ and $\{y_{te},\xbold_{te}\}$ are drawn from \textit{different marginal distributions} --- $\{P_{\xbold},\;P_{y|\xbold}\}$ and $\{P^{te}_{\xbold},\;P^{te}_{y|\xbold}\},$ respectively. An example for this scenario is when Kriging is used for predicting extrapolated spatial points with respect to the sample $T.$ It may happen due to sampling challenges, such as sampling from mountainous and deep marine regions \citep{rabinowicz2018assessing}. This scenario, which violates the setting that is assumed in Theorems \ref{wcv for same s} and \ref{EDF Theorem} requires further research.

\subsubsection{Longitudinal Data}\label{longitudinal data}
Another common setting with correlation structure is longitudinal data --- where there are several subjects that are repeatedly observed over time. In this setting, due to the temporal orientation, the correlation structure is more complicated than in a simple clustered data. 

Let us consider three scenarios that are equivalent to the three scenarios given in Section \ref{Kriging}. However, unlike in Section \ref{Kriging}, the multivariate framework would be considered (rather than the functional framework).

The first scenario is when the prediction goal is to predict a new observation, $T_{te},$ of one of the subjects in $T,$ sampled at a random time-point from the same distribution that the time-points of the observations in $T$ follow. Since $T_{te}$ and $T$ are sampled from the same subjects, then $\sbold_{te}=\sbold_{0}$ and therefore $P_{y_{te}|\xbold_{te},T_{-k}}= P_{y_{k}|\xbold_{k},T_{-k}}.$ By Theorem \ref{wcv for same s} this gives $w_{cv}=0,$ and therefore CV is suitable in this case.

The second scenario is when the prediction goal is predicting a new observation that relates to a subject that is not in $T$ and therefore $\sbold_{te}\neq \sbold_{0}$ and $P_{y_{te}|\xbold_{te},T_{-k}}\neq P_{y_k|\xbold_{k},T_{-k}}.$ For this scenario, given that $\{y_k,\xbold_{k}\}$ and $\{y_{te},\xbold_{te}\}$ were sampled from the same marginal distribution, $\{P_{\xbold},\;P_{y|\xbold}\},$ and a linear model is implemented, then by Theorem \ref{EDF Theorem}, $\widehat{CV_c}$ should be used instead of CV.  Numerical analysis of $\widehat{CV_c}$ implementation in this scenario is presented in Section \ref{simulation}.

Another scenario is when
$\{y_k,\xbold_{k}\}$ and $\{y_{te},\xbold_{te}\}$ are sampled from different marginal distributions, $\{P_{\xbold},P_{y|\xbold}\}$ and $\{P^{te}_{\xbold},\;P^{te}_{y|\xbold}\},$ respectively.
The marginal distributions can be different due to various reasonable prediction goals, such as forcing the data point in $T_{te}$ to extrapolates the data points in $T$ with respect to the time variable, which results in $\xbold_{te}$ and $\xbold_{k}$ being nonidentically distributed. This scenario violates the assumptions in Theorems \ref{wcv for same s} and \ref{EDF Theorem} and therefore requires further research.



\subsection{Comparison with Other Methods}\label{Comparison with other methods}
Several cross-validation variants were proposed for settings involving correlated data. Some of them were proposed from a perspective that correlation between the folds causes K-fold CV to underestimate the generalization error. As was shown above, this perception is wrong in many scenarios. Other variants are relevant for very specific applications under various sampling restrictions. Below, several cross-validation variants are described and compared to $\widehat{CV_c}.$ 



One method is h-blocking \citep{burman1994cross}, which is mainly relevant for spatial data. In h-blocking, in order to reduce the correlation between the folds, the analyzed surface is partitioned into blocks (folds) that are separated from each other by some distance, $h.$  As was described above, many use cases do not require any correlation reduction between the folds and K-fold CV is suitable, however the use of h-blocking is often recommended regardless of the predictive problems setup \citep{roberts2017cross}. Let us focus on a scenario when $P_{y_{te}|\xbold_{te},T_{-k}}\neq P_{y_{k}|\xbold_{k},T_{-k}},$ and therefore the condition in Theorem \ref{wcv for same s} is not satisfied, causing K-fold CV to be biased. 
In this scenario, although the h-blocking approach may seem reasonable, in fact, it suffers from several issues that do not affect $\widehat{CV_c}.$ For example, frequently, creating the separation between the folds requires omitting observations from the training sample. In addition, the folds that are generated by h-blocking have different distributions, in particular their distributions are different than $P_{T_{te}}.$ Therefore, h-blocking may provide a biased prediction error estimator with respect to the generalization error. Moreover, since some of the blocks are at the edge of the surface, then the prediction of those blocks becomes predicting spatial extrapolation, which might be inaccurate and does not reflect the planned prediction problem. This implication can affect dramatically the prediction error estimate \citep{roberts2017cross}.

Another method is leave cluster out \cite[LCO]{rice1991estimating}. This method is relevant for the case when $\Cov(\ybold_{tr},y_{te})=0$ and the training set, $T,$ has a clustered correlation structure. LCO eliminates the correlation between the folds by defining each cluster as a fold. This method suffers from several challenges that do not appear in $\widehat{CV_c}.$ First, using LCO forces the number of folds to be equal to the number of clusters. Another issue is when different clusters contain a substantially different number of observations, in which case LCO prediction error estimator can be biased with respect to the generalization error. In addition, validity of LCO is challenged when some clusters have different distribution than other clusters. In this case, as opposed to the generalization error definition, the observations in $X_{k}$ and $X_{-k}$ are nonidentically distributed.

Another cross-validation variant that is relevant for a balanced longitudinal data setting is leave observation from each cluster out \cite[LOFCO]{wu2002local}. This method is relevant for the case when the goal is to predict a new observation for each one of the subjects that appear in the training set. The folds partitioning mechanism in LOFCO is that different folds refer to different time-points, such that each fold contains the observations that were collocated at the same time-point across all the subjects. This partitioning is feasible due to the balanced data design assumption. The challenges in this method are similar to the challenges mentioned above: it requires a balanced data design, the number of folds are forced by the data structure, $X_{k}$ and $X_{-k}$ are not identically distributed --- in particular their time-points covariate is nonidentically distributed.

LCO, LOFCO and h-blocking reflect the understanding that nonstandard CV may be needed in presence of correlations, and offer solutions for very specific types of datasets and correlation settings that apply to any modeling technique. In contrast, $\widehat{CV_c}$ can be applied in a wide range of types of datasets and correlation settings, however it is limited to linear models.   


One last estimator that should be compared to $\widehat{CV_c}$ is 
\begin{align}\label{Altman}
\frac{\frac{1}{n}\|\ybold-H\ybold\|^2}{\Big[1-\frac{1}{n}\tr\Big(H\Cov\big(\ybold,\ybold\big)\Big)\Big]^2},
\end{align}
where $H$ satisfies, $\hat{\ybold}(X;T)=H\ybold.$ This estimator has the same spirit as generalized CV \cite[GCV]{craven1978smoothing} which approximates LOO for i.i.d. data as follows:
\[
GCV=\frac{\frac{1}{n}\|\ybold-\hat{\ybold}(X;T)\|^2}{\Big[1-\frac{1}{n}\tr\big(H\big)\Big]^2}.
\]
Eq. (\ref{Altman}) was first proposed by \cite{altman1990kernel} in context of time-series and later by \cite{opsomer2001nonparametric} in context of spatial data analysis. \citeauthor{altman1990kernel}'s motivating application was selecting the best bandwidth for modeling the trend in the data. He argues that replacing $\tr\big(H\big)$ in GCV by $\tr\Big(H\Cov\big(\ybold,\ybold\big)\Big)$ denoises the correlation between observations. Although Eq. (\ref{Altman}) and $\widehat{CV_c}$ are different and developed for different applications, they share some similarity in suggesting a corrected version for CV rather than controlling the partitioning scheme. Analyzing numerically Eq. (\ref{Altman}) did not yield promising results, in fact its approximated expectation did not converge to the same scale as the approximated generalization error, CV and $\widehat{CV_c}.$ We conclude that the heuristic argument behind the derivation of Eq. (\ref{Altman}) as an estimator of the generalization error does not hold in practical cases where the rigorous $\widehat{CV_c}$ gives valid results.

\subsection{Estimating \texorpdfstring{$\Cov(\ybold,\ybold)$}{TEXT}}\label{Estimating covariance}
In most applications $\Cov(\ybold,\ybold)$ is unknown in advance and therefore in order to implement $\widehat{CV_c},$ it should be estimated. In many cases, such as when GLS is implemented, the covariance matrix is already estimated in the model fitting part and can be also used in $\widehat{CV_c}.$ The use of estimated covariance matrix instead of the true covariance matrix can add variance to $\widehat{CV_c}$ and even can make $\widehat{CV_c}$ biased, especially when the estimation of the covariance matrix is imprecise, due to small sample size or other reasons. Several papers analyze the effect of plugging in the estimated covariance matrix instead of the true one in context of prediction error estimation, in various CV versions \citep{altman1990kernel,francisco2005smoothing} or AIC-type versions \citep{vaida2005conditional,liang2008note,greven2010behaviour,rabinowicz2018assessing}. Some papers derive extra corrections for their prediction error estimators which are theoretically valid in certain use cases, but are not applicable due to computational complexity \citep{liang2008note}, others show numerically that the prediction error estimator that are based on the estimated covariance matrices instead of the true covariance matrices performs well. Using simulations and real data analysis, we show in Section \ref{numerical part}, that the effect of using the estimated covariance matrices in $\widehat{CV_c}$ on its bias and variance is small, especially when the sample size of the training set is not very small. These results support the results of \cite{altman1990kernel,vaida2005conditional,francisco2005smoothing}.

\section{NUMERICAL RESULTS}\label{numerical part}
This section compares $\widehat{CV_c}$ and CV, with respect to the approximated generalization error, using simulation and real datasets analyses. Different prediction goals and correlation structures are analyzed.
Relevant datasets and code can be found in
\href{https://github.com/AssafRab/CVc}{https://github.com\\/AssafRab/CVc}.
\subsection{Simulation}\label{simulation}
The dependent variable, $\ybold\in\mathbb{R}^{n},$
was sampled from the following model:
\begin{gather*}
y_{i,j,k}=0.1\sum_{r=1}^{9} x_{i,j,k,r}+u_i+\sum_{r=1}^{2} z_{i,j,k,r}b_{i,j,r}+\epsilon_{i,j,k},\\
i\in\{1,...,I\},\;j\in\{1,...,5\},\;k\in\{1,...,10\},
\end{gather*}
where
\begin{itemize}
    \item the random effects, $u_i,\;\bbold_{i,j}=[b_1,b_2]$ and $\epsilon_{i,j,k}$ are independent and distributed as follows:
    \begin{gather*}
u_i\overset{ind}{\sim} N\big(0,3^2\big),\;\bbold_{i,j}=[b_{i,j,1},b_{i,j,2}]\in\mathbb{R}^2\overset{ind}{\sim} N(0,\begin{bmatrix} 3^2 & 0\\ 0 & 1 \end{bmatrix}),\;\epsilon_{i,j,k}\overset{ind}{\sim} N\big(0,1\big),
\end{gather*}
\item the covariates, $\{\xbold_r\}_{r=1}^{9}$ are:
\begin{itemize}
    \item $x_{i,j,k,1}=1,\;x_{i,j,k,2}=k,\;\forall i,j,k$ are the intercept and the time covariates,
    \item 
$x_{i,j,k,r}=\eta_i+\delta_{i,j,k},\;\forall i,j,k$ and $\forall r\in\{3,...,9\},$
where $\eta_i\overset{ind}{\sim} N(0,1),\;\delta_{i,j,k}\overset{ind}{\sim} N(0,1)$ are independent,
\end{itemize}
\item $z_{i,j,k,1}=1,\;z_{i,j,k,2}=k\;\forall i,j,k$ are the covariates for random intercept and random slope.
\end{itemize}


This settings was simulated for three different sample sizes, $n=300/400/500$ (where $I$ varies respectively to $6/8/10$).


This is a hierarchical clustered structured with $I$ clusters of $50$ observations each, where within each cluster there are five subclusters of ten observations each. There is a random intercept for the high-level clusters and random intercept and slope for the subclusters. Therefore the covariance of $\ybold|\xbold_{1},...,\xbold_{9}$ is:
\[
\Cov(y_{i,j,k},y_{i',j',k'})=9\times1_{[i=i']}+1_{[i=i',j=j']}(9+1\times k\times k')+1_{[i=i',j=j',k=k']}.
\]

GLS was fitted using LOO  for eight nested models. Model $1$ contains the intercept and time, Model $2$ also contains $\xbold_{3},$ and so on. Model $8$ contains all the covariates.
\
\subsubsection{Estimating Prediction Error}\label{estimating prediction error, simulation}
$T_{tr}$ is a random subset of $T$ of size $n-1.$ $T_{te}$ is a single observation, drawn from the same marginal distribution, however with new independent realizations of all the random effects. Therefore, when implementing CV, while in $T:$
\[\Cov(\ybold_{-k},y_{k})\neq0\;\;\forall\;k\in\{1,...,n\},\]
in $\{T_{tr},T_{te}\}:$
\[\Cov(\ybold_{tr},y_{te})=0.\]
\vspace{1em}


The generalization error was approximated by averaging $\big(y_{te}-\hat{y}(\xbold_{te};T_{tr})\big)^2,$ based on
$1000\times n$ samples of $\{T_{te},T_{tr}\}.$ The densities of CV and $\widehat{CV_c}$ were approximated based on $1000$ samples of $T.$ Since in this setting $\Cov(\ybold_{tr},y_{te})=0,$ the correction in this case is $2/n\times\tr\big(H_{cv}\Cov(\ybold,\ybold)\big).$

Figure \ref{density} shows the distribution of CV and $\widehat{CV_c}$ including their means for the saturated model when $n=400,$ compared to the generalization error. 
Two versions
of CV and $\widehat{CV_c},$  are presented --- when the variance parameters are known, and when they are estimated. 
\begin{figure}[t]
    \centering
      \includegraphics[width=0.7\linewidth]{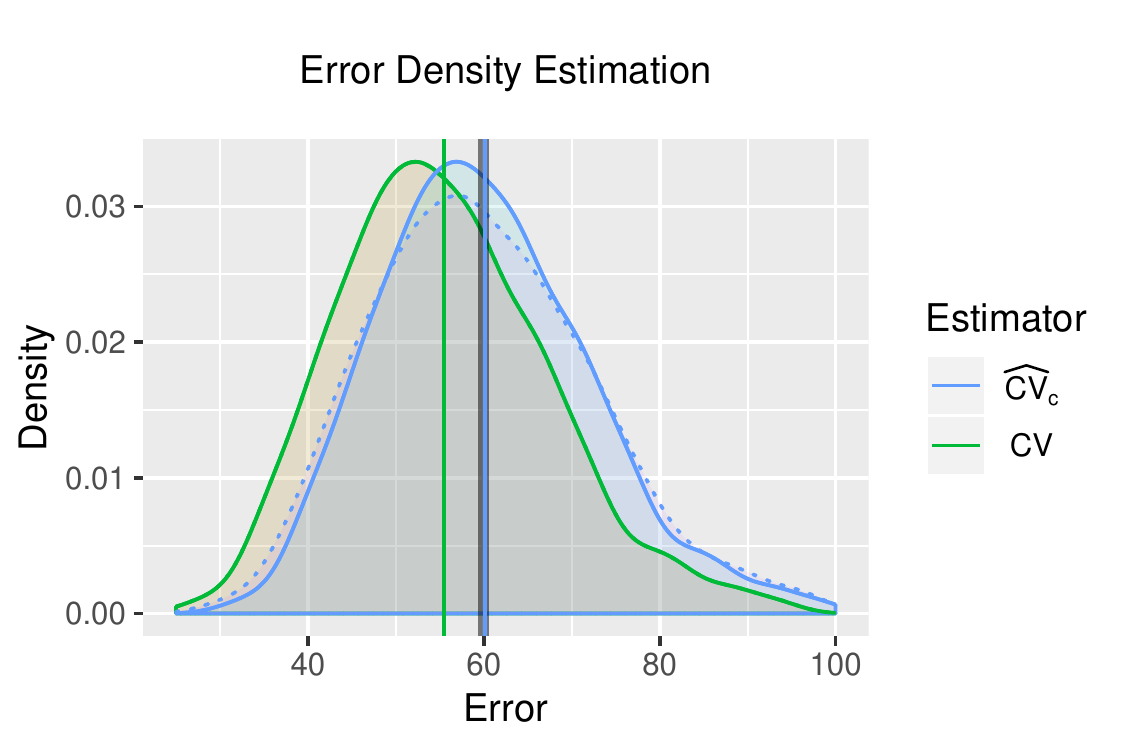}
    \caption{Densities of CV and $\widehat{CV_c}$ and their means compared to the generalization error (in black). Two scenarios are presented: when the variance parameters are known (solid line) and
when they are estimated (dashed line).}\label{density}
\end{figure}

As can be seen, $\widehat{CV_c}$ is an unbiased estimator of the generalization error, while CV estimator is biased. This is true for both versions, when the variance parameters are known and when they are estimated. Also, as expected, estimating the variance parameters increases the variance of $\widehat{CV_c}.$ Still, the density of the $\widehat{CV_c}$ version with the estimated variance parameters is similar to the version with the known variance parameters --- their averages are $60.14$ and $60.09$ (compared to $60.00$ for the approximated generalization error), their standard deviations are $12.07$ and $13.02,$ respectively.

In order to asses the performance of $\widehat{CV_c}$ version with the estimated variance parameters, compared to the version with the known variance parameters, a two sample
Anderson-Darling test \citep{anderson1952asymptotic} was used. The tested statistic is:
\[
\widehat{CV_c}-generalization\;error,
\]
where one sample uses $\widehat{CV_c}$ version with the estimated variance parameters, and the other
sample uses $\widehat{CV_c}$ version with the true variance parameters. Implementing the function
ad.test of the package kSamples in R software, the p-value of the test is $0.33.$ The result indicates that in this setting there is no evidence for significant difference between
the distribution of $\widehat{CV_c}-generalization\;error$ when the variance parameters are known in
advance or estimated.

Figure \ref{densityall} shows the same analysis as in Figure \ref{density}, however for various sample sizes.


\begin{figure}[H]
    \centering
      \includegraphics[width=0.7\linewidth]{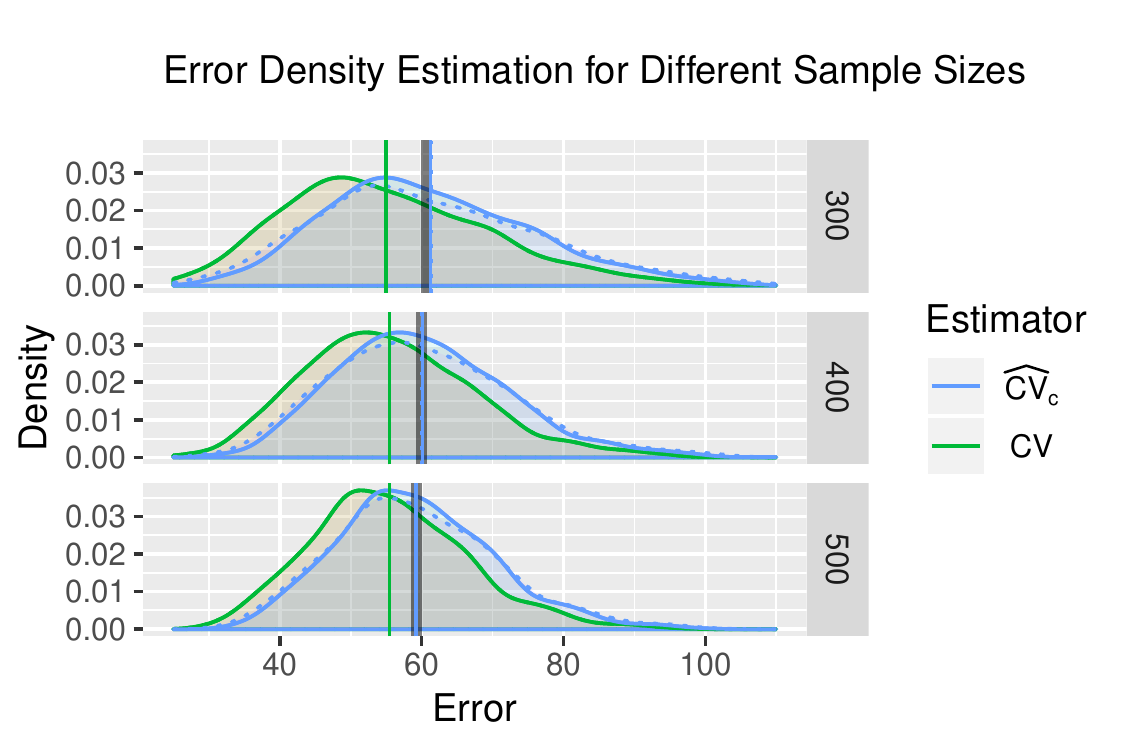}
    \caption{Densities of CV and $\widehat{CV_c}$ and their means compared to the generalization error (in black) for different sample sizes. Two scenarios are presented: when the variance parameters are known (solid line) and
when they are estimated (dashed line).}\label{densityall}
\end{figure}
As can be seen in Figure \ref{densityall}, the variance of the two $\widehat{CV_c}$ versions becomes similar as the sample size increases. While for $n=300$ the standard deviations of the version with the true variance parameters is $14.43$ and for the estimated is $16.04,$ when $n=500$ they are $11.25$ and $11.93,$ respectively. This is expected, since larger sample sizes provide more accurate variance parameters estimators. Also, as the sample size increases the CV bias, $w_{cv},$  decreases. This is specific to our setting, where $\Cov(\ybold,\ybold)$ becomes more sparse as the sample size increases and therefore $\E \big[2/n\times\tr\big(H_{cv}\Cov(\ybold,\ybold)\big)\big]$ decreases.

\subsubsection{\texorpdfstring{$\Cov(\ybold_{tr},y_{te})\neq0$}{TEXT} and Model Selection}\label{model selection, simulation}
In this simulation $\widehat{CV_c}$ is analyzed under a different prediction goal than in Section \ref{estimating prediction error, simulation}: estimating the generalization error of new observations from the \textit{same} high-level clusters, but from a \textit{different} subcluster, i.e.,  
\[
y_{te}=0.1\sum_{r=1}^{9}x_{te,r}+u_{i}+\sum_{r=1}^{2}z_{te,r} b_{te,r}+\epsilon_{te},
\]
where $u_{i}$ is a random effect realization that appears in the training sample, and $\bbold_{te}=[b_{te,1},b_{te,2}]^t$ is a new random effect realization vector. The new prediction goal derives a different correction --- since in this setting $y_{te}$ is correlated with $\ybold_{tr}$ through the random effect $u,$ then the correction in this case is $2/n\times\tr\big(H_{cv}\Cov(\ybold,\ybold|u_1,u_2...,u_I)\big),$ for more details see Section \ref{Specifying the correction}. 

Two different predictive algorithms were implemented:
\begin{itemize}
    \item LMM, which is the model that should be used in this setting under the normality assumption, since it utilizes explicitly the correlation between $y_{te}$ and $\ybold,$ using the relevant estimated random effects realizations, $\hat{u}_{1},...,\hat{u}_{I}.$ 
    \item GLS, which is inferior to LMM under the normality assumption, however unlike LMM, GLS is relevant for use cases where distributional assumptions cannot be taken. This is since GLS can be interpreted as an extension of the least squares algorithm, which does not rely on distributional assumptions.
\end{itemize}

Figure \ref{DensityLMMGLS} shows the model density estimation of $\widehat{CV_c}$ and CV compared to the approximated generalization error in the new setting when LMM and GLS are implemented.
\begin{figure}[t]
    \centering
    \begin{subfigure}[t]{0.5\textwidth}
        \centering
        \includegraphics[width=1\linewidth]{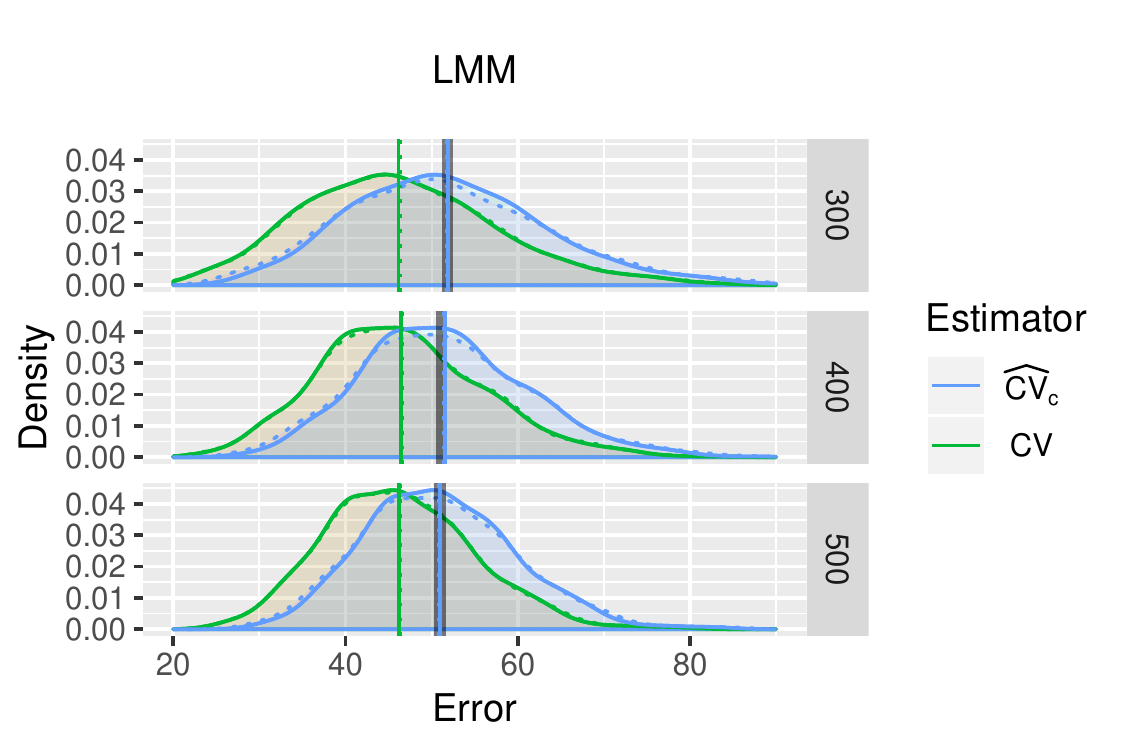}
        \caption{}\label{densityLMM}
    \end{subfigure}%
    ~ 
    \begin{subfigure}[t]{0.5\textwidth}
        \centering
        \includegraphics[width=1\linewidth]{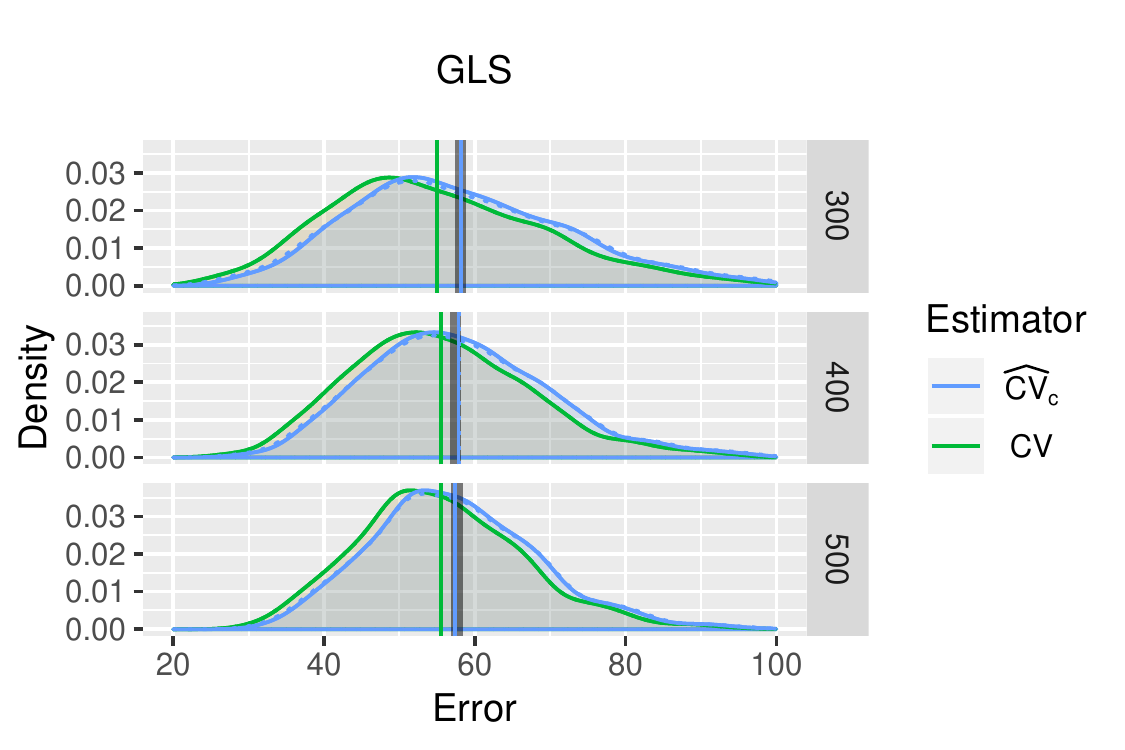}
        \caption{}\label{densityGLS}
    \end{subfigure}
    \caption{(\ref{densityLMM}): Using LMM, each curve refers to the density of the relevant estimator, CV and $\widehat{CV_c}.$ Two scenarios are presented: when the variance parameters are known (solid line) and
when they are estimated (dashed line). Their means are compared to the generalization error (in black). (\ref{densityGLS}): Similarly, the analysis is presented when GLS is implemented instead of LMM.}\label{DensityLMMGLS}
\end{figure}

As can be seen in Figure \ref{DensityLMMGLS}, $\widehat{CV_c}$ estimates the generalization error unbiasedly, while CV underestimates it. The CV bias in Figure \ref{densityGLS} is relatively small with respect to the biases in Figure \ref{densityLMM} and Figure \ref{densityall}:
\begin{itemize}
    \item The reason that the bias in Figure \ref{densityGLS} is smaller than the bias in Figure \ref{densityall} although in both settings GLS is implemented, is that the deviation of 
$P_{y_{te}|\xbold_{te},T_{tr}}$ from $P_{y_{k}|\xbold_{k},T_{-k}}$ is smaller in the setting of Figure \ref{densityGLS}, as expressed in the corrections --- $2/n\times\tr\big(H_{cv}\Cov(\ybold,\ybold|u_1,u_2...,u_I)\big)$ in the setting of Figure \ref{densityGLS} and $2/n\times\tr\big(H_{cv}\Cov(\ybold,\ybold)\big)$ in the setting of Figure \ref{densityall}. 
\item The reason that the bias in Figure \ref{densityGLS} is smaller than the bias in Figure \ref{densityLMM} although the deviation of 
$P_{y_{te}|\xbold_{te},T_{tr}}$ from $P_{y_{k}|\xbold_{k},T_{-k}}$ is the same in both settings is that in Figure \ref{densityLMM} LMM is implemented and therefore the correlation setting is utilized more than in Figure \ref{densityGLS}. While in GLS the realization of $\bbold$ affects the fixed effects estimates (see Section \ref{Interpretation of the Results}), in LMM the realization of $\bbold$ affects both, the fixed effects estimates and the estimates of $u_1,...,u_{I}.$
\end{itemize}

For demonstrating the model selection performance, Figure \ref{RateLMMGLS} presents the agreement
rates of $\widehat{CV_c}$ and CV with the approximated generalization error, over the repeated simulation runs. The results are presented for different sample sizes and for different predictive models (LMM and GLS). 
\begin{figure}[t]
    \centering
    \begin{subfigure}[t]{0.5\textwidth}
        \centering
        \includegraphics[width=1\linewidth]{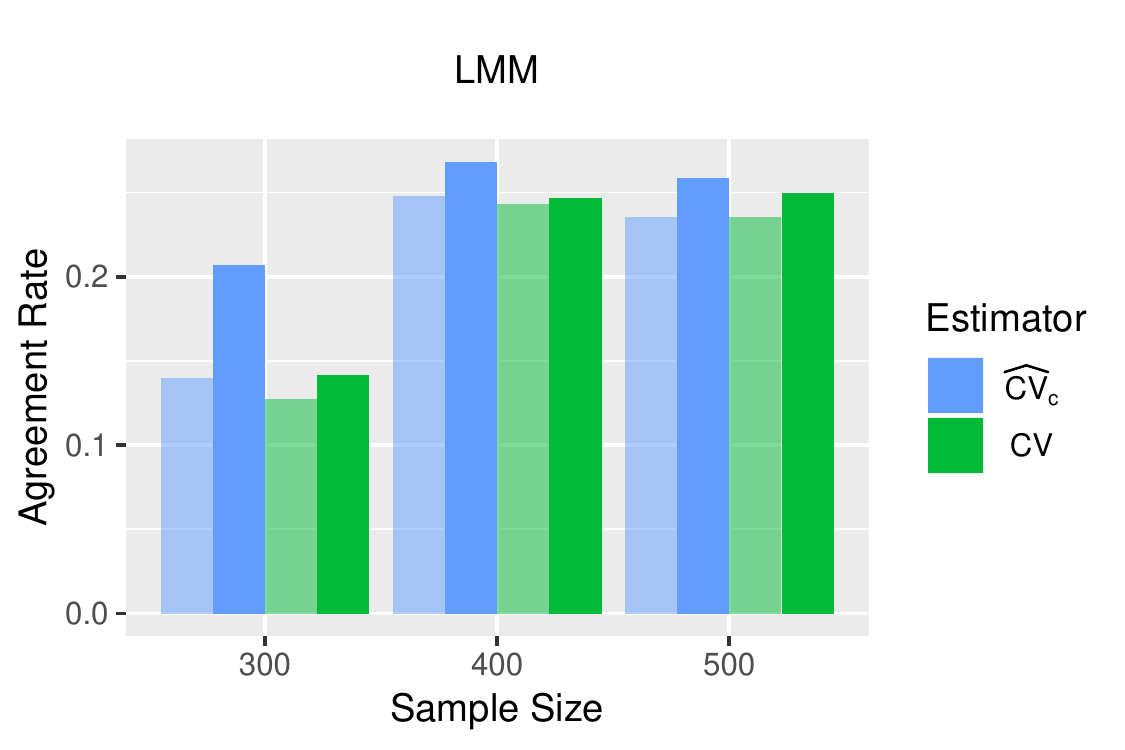}
        \caption{}\label{rateLMM}
    \end{subfigure}%
    ~ 
    \begin{subfigure}[t]{0.5\textwidth}
        \centering
        \includegraphics[width=1\linewidth]{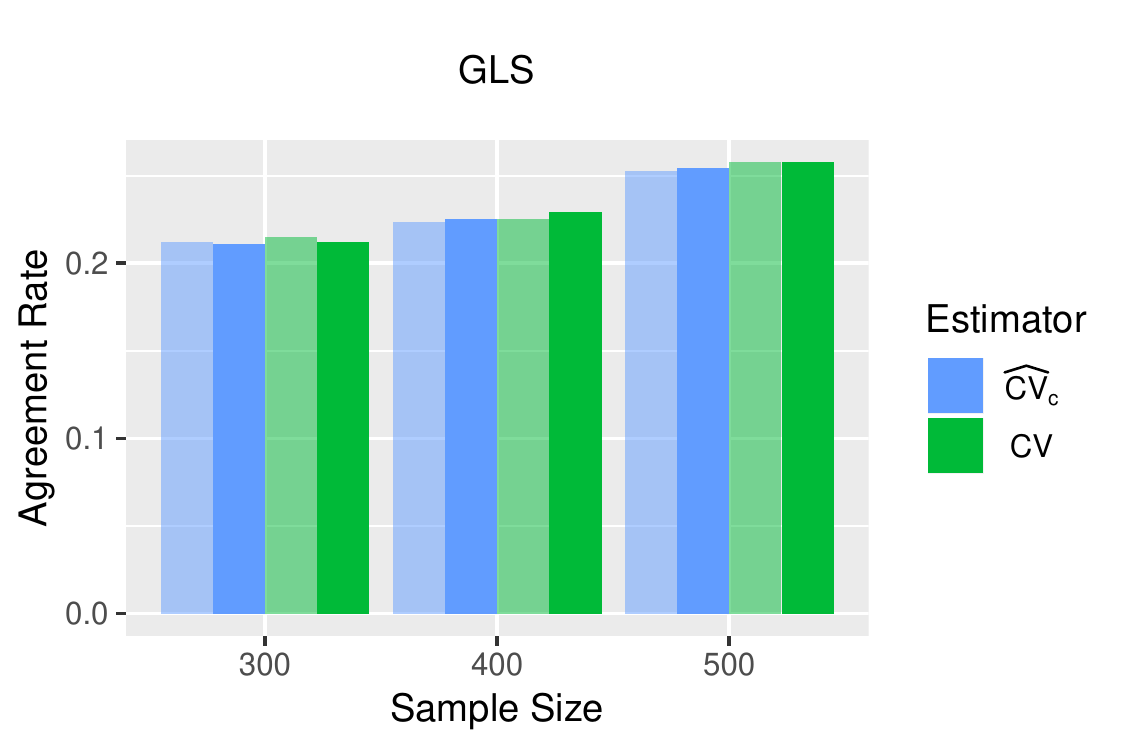}
        \caption{}\label{rateGLS}
    \end{subfigure}
    \caption{(\ref{rateLMM}): Each bar refers to the agreement rate of the relevant criterion, $\underset{h\in \mathcal{H}}{\argmin}\;\widehat{CV_c}(h)$ and $\underset{h\in \mathcal{H}}{\argmin}\;CV(h),$ with $\underset{h\in \mathcal{H}}{\argmin}\;generalization\;error(h)$ for different sample sizes. Also, two scenarios are presented:
when the variance parameters are known (in dark color) and when they are estimated (in light color). (\ref{rateGLS}): Similarly, the analysis is presented when GLS is implemented instead of LMM.}\label{RateLMMGLS}
\end{figure}


As can be seen from Figures \ref{RateLMMGLS}, when LMM is implemented $\underset{h\in \mathcal{H}}{\argmin}\;\widehat{CV_c}(h)$ performs better than $\underset{h\in \mathcal{H}}{\argmin}\;\widehat{CV}(h),$ however when GLS is implemented both criteria perform the same. It is important to note that in order to have a high agreement rate with the oracle, estimating unbiasedly the generalization error, as $\widehat{CV_c}$ does, is not enough since the variance of the estimator can mix the ranks of the models relative to the oracle. Therefore, there is no guarantee that $\underset{h\in \mathcal{H}}{\argmin}\;\widehat{CV_c}(h)$ performs better than $\underset{h\in \mathcal{H}}{\argmin}\;\widehat{CV}(h).$ This is also applies to other model selection criteria that are based on comparing the models' estimated prediction errors. However, in practice, as we see here and in the real dataset examples in the next section, $\underset{h\in \mathcal{H}}{\argmin}\;\widehat{CV_c}(h)$ performs better than $\underset{h\in \mathcal{H}}{\argmin}\;\widehat{CV}(h)$ when the CV bias is large. 

In addition, for LMM, both $\underset{h\in \mathcal{H}}{\argmin}\;\widehat{CV_c}(h)$ and $\underset{h\in \mathcal{H}}{\argmin}\;\widehat{CV}(h)$ with the true variance parameters perform better than the versions with the estimated variance parameters, while for GLS both perform similarly. This can be explained by the key role of the variance in predicting using LMM, compared to predicting using GLS. Also, when the sample size is small ($n=300$), it is more difficult to estimate accurately the variance parameters, and the difference between the two $\underset{h\in \mathcal{H}}{\argmin}\;\widehat{CV_c}(h)$  versions is substantial.






\subsection{Real Data Analysis}\label{real data analyses}
This section presents analysis of two real datasets, a dataset with a clustered correlation structure, and a dataset with hierarchical spatial correlation structure.

\subsubsection{Black Friday Dataset --- Clustered Correlation Structure}\label{black friday dataset, example}
The Black Friday dataset\footnote{Presented by \href{https://datahack.analyticsvidhya.com/contest/black-friday/}{Analytics Vidhya}.} contains information of $737,577$ transactions made in a retail store on Black Friday by $5,891$ customers. The median number of transactions by a customer is $53,$ and the range is $[5,1025].$ The dataset is available in \href{https://github.com/AssafRab/CVc}{https://github.com/\\AssafRab/CVc}.


A training sample of all transactions from $100$ random customers was drawn
and the goal is to fit a predictive model that predicts the purchase amount of a new transaction of a new random customer. Therefore, while in CV setting $\Cov(\ybold_{-k},y_{k})\neq0,$ in the prediction goal setting $\Cov(\ybold_{tr},y_{te})=0.$ Therefore the condition of Theorem \ref{wcv for same s} is not satisfied $(P_{T_{te},T_{tr}}\neq P_{T_{k},T_{-k}})$ and K-fold CV is unsuitable.

Three GLS models were fitted using different covariates (see Table \ref{Model Covariates BF}) with the covariance matrix $\Cov(y_{i,j},y_{i',j'})=\sigma^2_{customer}1_{[i=i']}+\sigma^2_{\epsilon}1_{[i=i',j=j']},$ where $i$ is customer index and $j$ is the index for an observation in a customer's set of observations. $\sigma^2_{customer}$ and $\sigma^2_{\epsilon}$ were estimated using restricted maximum likelihood of normal distribution.


CV and $\widehat{CV_c}$ were calculated using the training sample. Test error was calculated by averaging the prediction error of observations of other random $900$ customers using GLS models that were fitted by $T.$ This procedure was repeated $50$ times with different random training and test sets, and the generalization error was approximated by averaging the test errors of the $50$ runs. 

Figure \ref{BlackFriday} presents the means of CV, $\widehat{CV_c},$ and the test error over the $50$ runs, as well as the confidence intervals of the means.

\begin{figure}[ht!]
\centering
\begin{minipage}{0.6\textwidth}
\centering
     \includegraphics[width=1\linewidth]{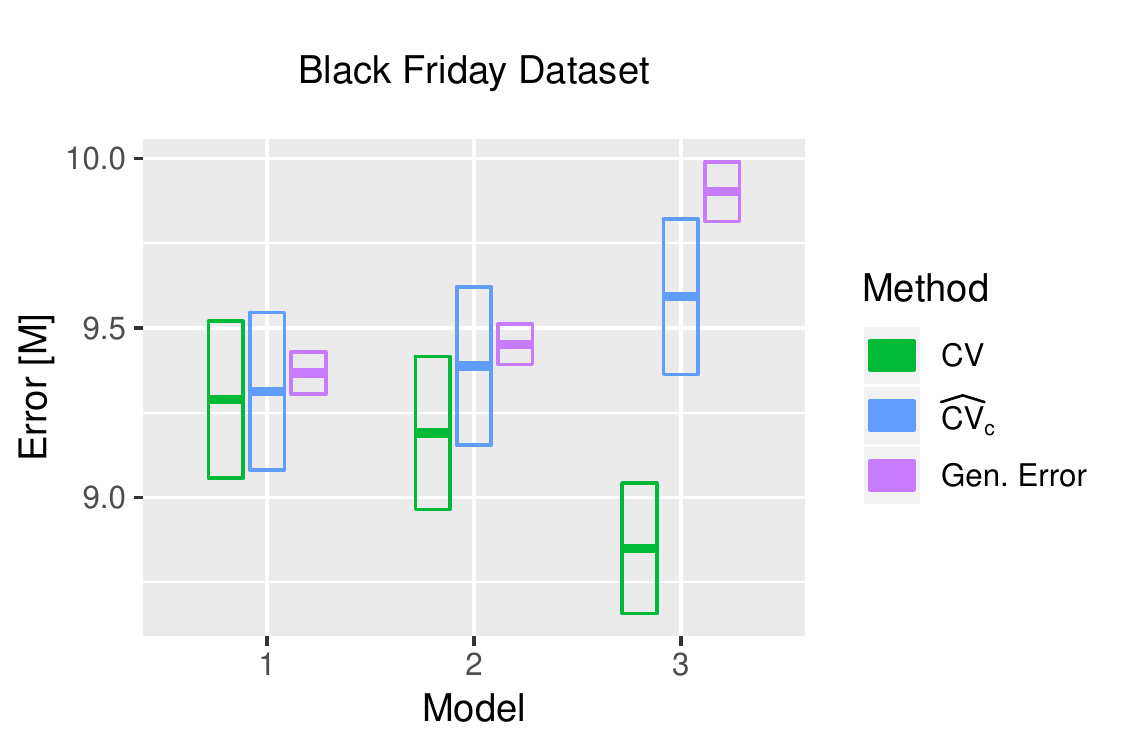}
    \caption{Prediction error estimation. For each model, the purple horizontal line in the box is the approximated generalization error, which was calculated by averaging the test error over the $50$ runs. The  box range is the two standard deviations confidence interval of the mean.
    Similarly, the mean estimated prediction error and the two standard deviations confidence interval of CV and $\widehat{CV_c}$ are presented.}\label{BlackFriday}
\end{minipage}
\hspace{1em}
\captionsetup{type=table}
\begin{minipage}{0.3\textwidth}
\resizebox{\textwidth}{!}{
\centering
\captionsetup{type=table} 
\begin{tabular}{|l|c|c|c|}
\hline 
 \multicolumn{1}{|c|}{Covariate\textbackslash{}Model}    & Model 1 & Model 2 & Model 3\tabularnewline
\hline 
\hline 
Intercept & $\checkmark$ & $\checkmark$ & $\checkmark$\tabularnewline
\hline 
Product Category 1 & $\checkmark$ & $\checkmark$ & $\checkmark$\tabularnewline
\hline 
Marital Status &  & $\checkmark$ & $\checkmark$\tabularnewline
\hline 
Gender &  & $\checkmark$ & $\checkmark$\tabularnewline
\hline 
Age &  & $\checkmark$ & $\checkmark$\tabularnewline
\hline 
Occupation &  &  & $\checkmark$\tabularnewline
\hline 
City Category &  &  & $\checkmark$\tabularnewline
\hline 
Stay in Current City Years &  &  & $\checkmark$\tabularnewline
\hline 
\end{tabular}
}
\vspace{3em}
\caption{Model 1 contains two covariates, model 2 contains five covariates and model 3 contains eight covariates.}\label{Model Covariates BF}
\end{minipage}
\end{figure}

As can be seen in Figure \ref{BlackFriday}, $\widehat{CV_c}$ estimates the generalization error better than CV for all the three models. Also, unlike CV, $\widehat{CV_c}$ selects the same model as the oracle, Model 1, and follows properly the prediction error estimation trend across the models. Analyzing the business aspect of this example, leads to the surprising conclusion that including only the product category as a fixed effect gives the best prediction model, and adding covariates like age and gender does not improve predictive performance, when the data are properly analyzed using $\widehat{CV_c}.$

\subsubsection{California Housing Dataset --- Clustered Random Field Correlation Structure}\label{California housing, example}
Another dataset that was analyzed is 
the California housing dataset, which contains house values and other housing parameters in California. The dataset has a hierarchical spatial correlation structure --- each observation has spatial coordinates value, where some of the observations share the same coordinates value and therefore define a cluster. This correlation structure is similar to the correlation structure in Sections \ref{simulation}, however with spatial data rather than longitudinal data. The dataset is available in the python scikit-learn package and the code is available in \href{https://github.com/AssafRab/CVc}{https://github.com/AssafRab/CVc}.


To emphasize the hierarchical clustering structure, all the clusters containing only a single observation were excluded from the analysis ($8,237$ observations out of  $20,640$), so the analyzed dataset contains $12,403$ observations from $4,353$ clusters. The median number of observations in a cluster is $2$ and the range is $[2,15].$

A training sample, $T,$ of $700$ independent clusters was randomly drawn. The prediction goal is to predict the 'median house value in a block' of a new observation from a \textit{different} cluster than the clusters in $T.$ Three GPR models with different covariates were fitted (see Table \ref{Model Covariates CA}) with the following covariance function:
\[
\Cov(y_{i},y_{i'})=K_{exp}(\zbold_{i}-\zbold_{i'})+\sigma^2\times1_{[i=i']},
\]
where $\zbold_{i}\in\mathbb{R}^{2}$ contains the latitude and longitude of observation $i$ and $K_{exp}(\cdot)$ is the exponential kernel function. The parameters of $K_{exp}(\cdot)$ and $\sigma^2$ were fitted using maximum likelihood of normal distribution. 


Since the prediction goal is to predict a new observation from a different cluster than the clusters in $T,$ then $\Cov(\ybold_{-k},y_{k})\neq\Cov(\ybold_{-k},y_{te}).$ Therefore the condition of Theorem \ref{wcv for same s} is not satisfied and standard K-fold CV is unsuitable. It is important to note that unlike the Black Friday example, here $\Cov(\ybold_{-k},y_{te})\neq0$ although $y_{te}$ is from a new cluster. This is due to the spatial correlation structure of the dataset, i.e., the latent random function $\sbold.$ The $\widehat{CV_c}$ correction in this case is $
2/n\times\tr\big(H_{cv}\Cov(\ybold,\ybold|\sbold)\big),$ where $\Cov(\ybold,\ybold|\sbold)$, which expresses the extra correlation that the clusters contribute over the spatial correlation of the data (see Section \ref{Specifying the correction}), is estimated as follows:
\[
\widehat{\Cov}(\ybold,\ybold|\sbold)[i,i']=
\begin{cases}
\begin{array}{c}
K_{exp}(0)-\frac{1}{|j,j'\; \forall\zbold_{j}\neq \zbold_{j'}|}\sum_{j,j'\;\forall\zbold_{j}\neq \zbold_{j'}}K_{exp}(\zbold_{j}-\zbold_{j'})\\
0
\end{array} & \begin{array}{c}
\zbold_{i}=\zbold_{i'}\\
\zbold_{i}\neq \zbold_{i'}\\
\end{array}\end{cases}.
\]

The analysis was repeated $50$ times with different random $T,$ at each run CV and $\widehat{CV_c}$ were calculated and the test error was calculated by averaging the prediction error of the remaining observations of the $3,653$ clusters that are not in $T$ using GPR models that were fitted by $T.$ Averaging the test error over the $50$ runs approximates the generalization error.

Figure \ref{California housing} presents the means of $CV,$ and $\widehat{CV_c},$ and the test error over the $50$ runs, as well as the confidence intervals of the means.

\begin{figure}[ht!]
\centering
\begin{minipage}{0.6\textwidth}
\centering
     \includegraphics[width=1\linewidth]{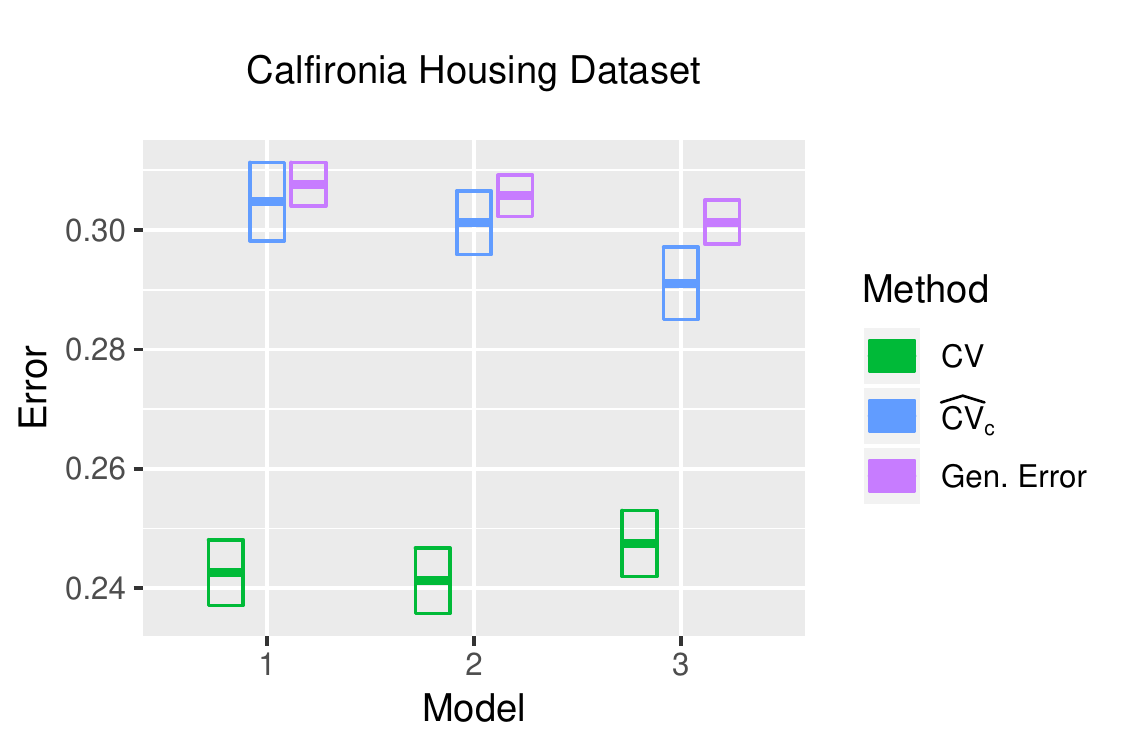}
    \caption{Prediction error estimation. For each model, the purple horizontal line in the box is the approximated generalization error, which was calculated by averaging the test error over the $50$ runs. The  box range is the two standard deviations confidence interval of the mean.
    Similarly, the mean estimated prediction error and the two standard deviations confidence interval of CV and $\widehat{CV_c}$ are presented. The unit of the dependent variable is in $\$100,000.$}\label{California housing}
\end{minipage}
\hspace{1em}
\captionsetup{type=table}
\begin{minipage}{0.3\textwidth}
  \resizebox{\textwidth}{!}{
\centering
\captionsetup{type=table} 
\begin{tabular}{|l|c|c|c|}
\hline 
 \multicolumn{1}{|c|}{Covariate\textbackslash{}Model}    & Model 1 & Model 2 & Model 3\tabularnewline
\hline 
\hline 
Intercept & $\checkmark$ & $\checkmark$ & $\checkmark$\tabularnewline
\hline 
Median income in block  & $\checkmark$ & $\checkmark$ & $\checkmark$\tabularnewline
\hline 
Median house age in block &  $\checkmark$ & $\checkmark$& $\checkmark$\tabularnewline
\hline 
Average number of rooms  &  & $\checkmark$ & $\checkmark$\tabularnewline
\hline 
Average number of bedrooms &  & $\checkmark$ & $\checkmark$\tabularnewline
\hline 
Block population &  &  & $\checkmark$\tabularnewline
\hline 
Average house occupancy &  &  & $\checkmark$\tabularnewline
\hline 
\end{tabular}
}
\vspace{2.2em}
\caption{Model 1 contains three covariates, model 2 contains five covariates and model 3 contains seven covariates.}\label{Model Covariates CA}
\end{minipage}
\end{figure}

Similarly to Figure \ref{BlackFriday}, $\widehat{CV_c}$ estimates the generalization error better than CV for all the models, selects the same model as the oracle (Model 3), and follows properly the prediction error estimation trend across the models. Also, while in Figure \ref{BlackFriday} CV underestimates the prediction error of the saturated model more than the smaller models, here CV underestimates the saturated model less than the smaller models. Therefore, biasedness of CV prediction error can be expressed in model selection in different ways --- sometimes by favoring over-parameterized models and sometimes by favoring under-parameterized models.

\section{Conclusion and Discussion}

In this paper we tackle the problem of applying CV as an estimate of generalization error in non-i.i.d. situations.
While the fundamental concerns that this presents are widely acknowledged, a clear understanding of when adjustments are needed, and what adjustments are appropriate, seems lacking in much of the literature \citep{roberts2017cross,saeb2017need,anderson2018comparing}. 

We first present a general formulation of the bias in using CV in presence of correlations,  
which leads to a clear general definition of settings where no correction to CV is needed (Theorem \ref{wcv for same s}). It shows that non-i.i.d. situations can still facilitate correctness of regular CV, as long as the dependence structure between training and prediction points is consistent in CV and actual prediction. This simple result appears to contradict some previous claims in the literature. An example can be found in \cite{roberts2017cross}, where mechanisms of controlling folds partitioning are suggested for Kriging tasks, based on the conception that CV is always over-optimistic for non-i.i.d. data.

We then present a derivation of a bias correction for linear models under general correlation structures (Theorem \ref{EDF Theorem}), which is used in establishing a bias corrected cross-validation version, $\widehat{CV_c}$. 
To implement our correction, it is necessary to specify the covariance structures within the training set and between the training and prediction sets. This is typically also required to choose a modeling approach. For example, in a simple linear mixed model with normal assumptions, if one assumes that the random effect realizations are the same when predicting, then LMM prediction is appropriate, while if random effect realizations are new, then using GLS  for estimation of fixed effects only for prediction is more appropriate \citep{verbeke1997linear}.  However, it is important to emphasize that the validity of the correction $\widehat{CV_c}$ does not depend on selection of an appropriate modeling approach. In other words, if one mistakenly uses GLS where LMM is appropriate, $\widehat{CV_c}$ still gives an unbiased estimate of generalization error for the resulting model, as long as the covariance structure is correctly specified.

In practice, the covariance matrices are typically not fully known, but partially estimated from the data (for example, variance parameters in LMM can be estimated using restricted maximum likelihood, \cite{verbeke1997linear}), and this is also required for applying $\widehat{CV_c}$ to correct CV results. 
This could potentially add uncertainty to the estimates, but as demonstrated in Section \ref{numerical part}, it does not tend to affect their expectation.

A fundamental assumption that is taken throughout this paper is that the marginal distributions of the training and prediction sets are the same. In case the marginal distributions are different, i.e., the observations in the training set are drawn from a different marginal distribution than the prediction data, then $\widehat{CV_c}$ is unsuitable. This scenario requires further research and relevant applications are given in Sections \ref{Kriging} and \ref{longitudinal data}. Other important use cases of different marginal distributions for training and prediction sets are discussed in \cite{rabinowicz2018assessing} and solutions for these use cases are proposed in context of transductive prediction error (rather than generalization error approach taken here).

An interesting situation also covered by $\widehat{CV_c}$ and not often discussed, is when the training set contains i.i.d. observations, but new data points where predictions are made are actually correlated with the training set (for example, new observations in the same set of cities). In this case, if the correlation structure is known, then $\widehat{CV_c}$ can still be used to estimate the prediction error.

\bibliography{Reference}
\bibliographystyle{JASA}

\end{document}